\pdfoutput=1
\documentclass[letterpaper,superscriptaddress,aps,pra,nolongbibliography,showpacs,floatfix,dvipsnames,10pt,twocolumn]{revtex4-2} 
\usepackage{url}
\usepackage{amsmath}
\usepackage{amssymb}
\usepackage[utf8]{inputenc}
\usepackage[english]{babel}
\usepackage[T1]{fontenc}

\usepackage{tikz}
\usetikzlibrary{arrows}

\usepackage{array}
\newcolumntype{P}[1]{>{\centering\arraybackslash}p{#1}}
\newcolumntype{M}[1]{>{\centering\arraybackslash}m{#1}}

\newcommand{\bbZ}{\mathbb{Z}}

\usepackage{amsmath}
\usepackage{hyperref}
\usepackage{lipsum}
\usepackage{graphicx,helvet}
\usepackage{color}
\usepackage{mathtools}
\usepackage{bbm,bm}
\usepackage{soul}
\usepackage{amsfonts}
 \usepackage{amsmath}
\usepackage{lipsum}
\definecolor{mygray}{gray}{0.4}
\definecolor{light-blue}{rgb}{0.8,0.85,1}
\graphicspath{{../images/}}
\usepackage{amsthm}

\usepackage[draft,inline,nomargin]{fixme} \fxsetup{theme=color}
\FXRegisterAuthor{cp}{acp}{\color{blue}CP}
\FXRegisterAuthor{fl}{afl}{\color{orange}FL}
\definecolor{jacolor}{RGB}{200,40,0} \FXRegisterAuthor{ja}{aja}{\color{jacolor}JA}
\FXRegisterAuthor{tb}{ttb}{\color{OliveGreen}TB}
\FXRegisterAuthor{af}{aaf}{\color{magenta}AF}

\mathchardef\Re="023C
\mathchardef\Im="023D

\newcommand{\jami}{Jamiołkowski}

\newcommand{\mcH}{\mathcal{H}}

\newcommand{\mcE}{\ensuremath{\mathcal{E} }}
\newcommand{\mcG}{\ensuremath{\mathcal{G} }}

\newcommand{\mcD}{\mathcal{D}}

\newtheorem{theorem}{Theorem}

\newcommand{\lpart}{\overline{L}}
\newcommand{\mpart}{\overline{M}}

\newcommand{\unit}[2]{U(#1,#2)}
\newcommand{\coef}[2]{\tau(#1,#2)}
\newcommand{\eref}[1]{eq.~(\ref{#1})} 
\newcommand{\sref}[1]{sec.~\ref{#1}}

\newcommand{\Fref}[1]{Fig.~\ref{#1}}

\newcommand{\tr}{\mathop{\mathrm{Tr}}u\nolimits}

\newcommand{\ket}[1]{{\vert #1 \rangle}}
\newcommand{\bra}[1]{{\langle #1 \vert}}

\renewcommand{\tr}{\mbox{\rm Tr}}

\usepackage{physics}
\usepackage{breqn}

\newcommand{\vm}{\vec m}
\newcommand{\vn}{\vec n}
\newcommand{\Uv}{U(\vm,\vn)}
\newcommand{\tauv}{\tau (\vm,\vn)}

\renewcommand\star{\ast} 

\newcommand{\unam}{Universidad Nacional Aut\'onoma de M\'exico, Ciudad de M\'exico 01000, Mexico}
\newcommand{\icf}{Instituto de Ciencias F\'{\i}sicas, Universidad Nacional Aut\'onoma de M\'exico, Cuernavaca 62210, Mexico}
\newcommand{\ifunam}{Instituto de F\'{\i}sica, \unam}

\newcommand{\ciencias}{Facultad de Ciencias, Universidad Nacional Autónoma de México, Ciudad de México 01000, Mexico}
\newcommand{\affalejandro}{Departamento de Física, CCEN, Universidade Federal de Pernambuco, Recife 50670-901, PE, Brazil}
\begin{document}
\title{Weyl channels for multipartite systems} 
\author{Tomás Basile} \affiliation{\ciencias}
\author{Jose Alfredo de Leon} \affiliation{\ifunam}
\author{Alejandro Fonseca} \affiliation{\affalejandro}
\author{François Leyvraz} \affiliation{\icf}
\author{Carlos Pineda} \email{carlospgmat03@gmail.com} \affiliation{\ifunam}
\begin{abstract}

Quantum channels, a subset of quantum maps, describe the unitary and non-unitary
evolution  of quantum systems.
We study a generalization of the concept of Pauli maps to the case of
multipartite high dimensional quantum systems through the use of the Weyl
operators. 
The condition for such maps to be valid quantum channels, i.e. complete
positivity, is derived in terms of Fourier transform matrices. 
From these conditions, we find the extreme points of this set of channels and 
identify an elegant algebraic structure nested \textit{within} them. 
In turn, this allows us to expand upon the concept of 
'component erasing channels' introduced in
earlier work by the authors.
We show that these channels are completely 
characterized by elements drawn of finite cyclic groups.
An algorithmic construction for such channels is presented and the smallest subsets of
erasing channels which generate the whole set are determined.
\end{abstract}
 
\maketitle

\newcommand{\sa}{\ensuremath{\textit{a}}}
\newcommand{\san}{\ensuremath{\textit{A}}}
\section{Introduction} 
The description of open quantum systems \cite{breuer2007theory,rivas2012open}
serves a twofold purpose. Firstly, it lies at the core of the measurement
problem~\cite{Schlosshauer2004,RevModPhys.75.715}, thus bearing a fundamental
interest. On the other hand, it describes quantum systems where the inevitable
interaction with an environment is taken into account
\cite{nielsen_chuang_2010}. For most implementations of quantum devices it is
crucial to understand and control such unwanted interaction.  In both cases, a
natural language for such a description is that of quantum channels, which have
been subject of intense research \cite{wilde2017}. 

The properties of a quantum channel dictate the characteristics of the
associated quantum dynamics. In the realm of qubits, the set of quantum
channels has been explored and thanks to a better understanding of its
geometry, several physical properties of the set, such as divisibility
\cite{Wolf2008assesing,Davalos2019,heinosaari2011mathematical},
non-Markovianity \cite{Rivas2014,Breuer2016}, channel capacity
\cite{Gyongyosi2018}, among others have been unraveled.  In a previous paper
\cite{DeLeon2022} we proposed and studied a class of channels acting on
multi-qubit systems that either erased or preserved the Pauli components of the
state. These are the so  called Pauli component erasing (PCE) maps, which are
an important subset of the Pauli maps.
We found that every PCE channel corresponds uniquely to a vector subspace of a
discrete vector space.  Such channels can be associated with measurements and
asymptotic Lindbladian evolution. 

Moreover, most of the applications in the field of quantum information have
been built upon qubits. Nevertheless, many real-world realizations of quantum
systems have more than two levels that can be used to provide an important
technical advantage. Such advantage is indeed employed to develop several
important tasks like quantum cryptography \cite{Brus2002,Cerf2002}, quantum
computation \cite{Ralph2007,Campbell2014,Wang2020}, violation of Bell
inequalities \cite{Vertesi2010}, randomness generation \cite{Skrzypczyk2018},
among others.  For this reason, the study of high-dimensional and multiparticle
systems is of relevance.

In this article, we introduce the concept of Weyl channels for systems composed
of many particles, allowing  each of these to be of different dimensions.  We
begin defining these channels in \sref{sec:weyl} as diagonal channels in the
basis of multi-particle Weyl matrices, which are tensor products of the
well-known Weyl matrices. Moving forward, we proceed to diagonalize the
Choi-\jami{} matrix, revealing a linear relationship between the eigenvalues
and those of the channel. From this, we find two significant properties of the
set of Weyl channels: (1) its extreme points in \sref{sec:extreme}, and (2) a
subgroup structure of all Weyl channels in \sref{sec:border}. Then, in
\sref{sec:WCE} we extend the notion of \textit{component erasing} channels by
introducing the Weyl erasing channels. Given its semigroup property, we
describe the generator subset by means of the aforementioned algebraic
structure of Weyl channels. Finally, we wrap up and conclude in
\sref{sec:conclusions}.
\section{Weyl channels}
\label{sec:weyl}
A well-known generalization of the Pauli matrices to arbitrary $d$-dimensional Hilbert spaces
was introduced by Weyl~\cite{weyl1927quantenmechanik} and involves the following unitary matrices~\cite{Bertlmann2008}:
\begin{equation}\label{eq:general1}
\unit{m}{n}=\sum_{k=0}^{d-1}\omega^{mk}\ket{k}\bra{k+n}.
\end{equation}
Here we introduce the notation we shall use throughout: $\omega$ is the primitive $d$-th root
of unity $\exp(2\pi i/d)$. 
All arithmetical operations over latin indices are taken over modulo $d$.

We will further be mainly concerned with systems of $N$ qudits, for which we introduce the following
standard notations: 
\begin{equation}
\unit{\vec m}{\vec n}=\mathop{\bigotimes}_{\alpha=1}^N\unit{m_\alpha}{n_\alpha}
\label{eq:general2}
\end{equation}
Greek indices will always run over a range  
from 1 to $N$, and the arithmetic operations over them will always be the usual ones. When the range is
not specified, it will be from 1 to $N$. 

We now write for example
\begin{equation}
\unit{\vec m}{\vec n}=\sum_{\vec k}\omega^{\vec m\cdot \vec k}\ket{\vec k}\bra{\vec k+\vec n},
\label{eq:general3}
\end{equation}
the notational conventions being self-explanatory. Note further that all our
results can routinely be extended to the more complicated case in which the
different particles in the $N$-particle system have different dimensions, $d_\alpha$. The ``vectors'' $\vec m$ are then replaced by lists of integers,
with $0\leq m_\alpha\leq d_\alpha-1$.
Whereas this complicates the notation
considerably, no points of essential interest are thereby introduced. We thus
leave it to the interested reader to develop these issues. When non-trivial
points arise in this respect, we shall explicitly point this out. 

These unitary matrices satisfy certain elementary properties:
\begin{align}
\tr \unit{m}{n}^\dagger\unit{m^\prime}{n^\prime}&=d\delta_{mm^\prime}\delta_{nn^\prime}
\label{eq:general4} \\
\unit{m}{n}\unit{m^\prime}{n^\prime}&=\omega^{m^\prime n}\unit{m+m^\prime}{n+n^\prime}
\nonumber \\
\unit{m}{n}\unit{m^\prime}{n^\prime}&=\omega^{m^\prime n-m n^\prime}\unit{m^\prime}{n^\prime}
\unit{m}{n}
\nonumber \\
\unit{m}{n}^\dagger&=\omega^{mn}\unit{-m}{-n} \nonumber
\end{align}
as well, of course, as their vectorial equivalents. 

We now define Weyl maps and the corresponding channels: any density matrix on the space of
$N$ qudits, that is, on $(\mathbb{C}^d)^{\otimes N }$, can be expressed as
\begin{equation}
\rho=\frac{1}{d^N}\sum_{\vec m,\vec n}\alpha(\vec m,\vec n)\unit{\vec m}{\vec n}
\label{eq:general5}
\end{equation}
where $\alpha(\vec m,\vec n)$ satisfies 
$\alpha(\vec m,\vec n)=\omega^{\vec m\cdot \vec n}\alpha^*(-\vec m,-\vec n)$
in order for $\rho$ to satisfy the condition of hermiticity. 
More intricate conditions need to be satisfied
in order to yield a positive matrix, 
but we shall not be concerned with these.

A Weyl map is now defined as follows 
\begin{equation}
\rho \rightarrow
\rho' = 
\mcE[\rho]=\frac{1}{d^N}\sum_{\vec m,\vec n}\coef{\vec m}{\vec n}\alpha(\vec m,\vec n)\unit{\vec m}{\vec n}.
\label{eq:general6}
\end{equation}
Here the $\coef{\vec m}{\vec n}$ are complex numbers, whereas $\rho$ is the density matrix given in 
(\ref{eq:general5}). 
In other words, if the $U(\vec m, \vec n)$ are viewed as generators of the
vector space of all Hermitian matrices, the Weyl maps act {\em diagonally\/} on
this set. 

We now wish to find the conditions necessary and sufficient for $\mcE$ to be a quantum channel, that
is, to be trace and hermiticity preserving, as well as completely positive. For the former two
conditions, we require
\begin{subequations}
\begin{eqnarray}
\coef{\vec m}{\vec n}&=&\coef{-\vec m}{-\vec n}^\star,
\label{eq:general7a}\\
\coef{0}{0}&=&1.
\label{eq:general7b}
\end{eqnarray}
\label{eq:general7}
\end{subequations}
To verify complete positivity, we must check the circumstances under which the
Choi--\jami{} matrix, given by
\begin{equation}
\mcD=\frac1{d^N}\sum_{\vec m,\vec n}\coef{\vec m}{\vec n}\unit{\vec m}{\vec
n}\otimes\unit{\vec m}{\vec n}^\star
\label{eq:general8}
\end{equation}
is positive semidefinite. Interestingly, \eref{eq:general8} is the
corresponding Choi--\jami{} matrix, even if $U(\vec m,\vec n)$ are not Weyl
operators, but an arbitrary basis of Hilbert-Schmidt space, as shown in
Appendix \ref{app:choi}.
%
%
To specify the criteria for matrix in \eref{eq:general8} to be
positive semidefinite, we evaluate its
eigenvalues $\lambda(\vec m,\vec n)$. This is easily done after noticing that the various
elements of the sum, namely the $\unit{\vec m}{\vec n}\otimes\unit{\vec m}{\vec n}^\star$
all commute for arbitrary values of $\vec m$ and $\vec n$, as readily follows from 
(\ref{eq:general4}):
\begin{widetext}
\begin{eqnarray}
\Big(\unit{m}{n}\otimes\unit{m}{n}^\star\Big)
\Big(\unit{m^\prime}{n^\prime}\otimes\unit{m^\prime}{n^\prime}^\star\Big)&=&
\Big(\unit{m}{n}\unit{m^\prime}{n^\prime}\Big)\otimes\Big(\unit{-m}{n}
\unit{-m^\prime}{n^\prime}\Big)
\nonumber\\
&=&\Big(\unit{m+m^\prime}{n+n^\prime}\Big)
\otimes\Big(\unit{-(m+m^\prime)}{n+n^\prime}\Big)
\label{eq:general8a}
\end{eqnarray}
\end{widetext}
The symmetry of the final expression proves the claim, and the extension to the case of arbitrary 
$N$ is straightforward. 

%
It now remains to determine the eigenvalues of $\unit{\vec m}{\vec n}$, which given its tensor product structure [see \eref{eq:general2}]
can be reduced to the single-qudit case of $\unit{m}{n}$. 
These can be calculated directly
studying the recursion relation that follows from the eigenvalue
equation 
for the Weyl operators, see Appendix~\ref{app:eigenvalues}.
%
One can then readily see that the eigenvalues $\mu(r,s)$ of
$U(m, n)\otimes U(m, n)^\star$ take the form
\begin{equation}
\mu(r,s)=\omega^{mr-ns},
\label{eq:general10v2}
\end{equation}
where $r$ and $s$ are arbitrary integers modulo $d$ that serve as labels for
the eigenvalue.  The degeneracy pattern
of these eigenvalues is complicated,
but since our
focus is on the positivity of $\mcD$, we do not need to consider these details. 

The set of eigenvalues of $\unit{\vec m}{\vec n}\otimes\unit{\vec m}{\vec n}^*$ is then given by
\begin{equation}
\mu(\vec r,\vec s)=\omega^{\vec m\cdot\vec r-\vec n\cdot\vec s}.
\label{eq:general11v2}
\end{equation}
The condition for the positive semidefiniteness of $\mcD$ is thus that, for all $\vec r$ and $\vec s$,
\begin{equation}
d^{-N}\sum_{\vec m,\vec n}\coef{\vec m}{\vec n}\omega^{\vec m\cdot\vec r-\vec n\cdot\vec s}=
\lambda(\vec r,\vec s)\geq0.
\label{eq:general12}
\end{equation}
Note that condition (\ref{eq:general7a}) on $\coef{\vec m}{\vec n}$ straightforwardly shows that the
left-hand side of (\ref{eq:general12}) is real, so that the inequality is meaningful. 

\begin{figure*}
\includegraphics[width=\textwidth]{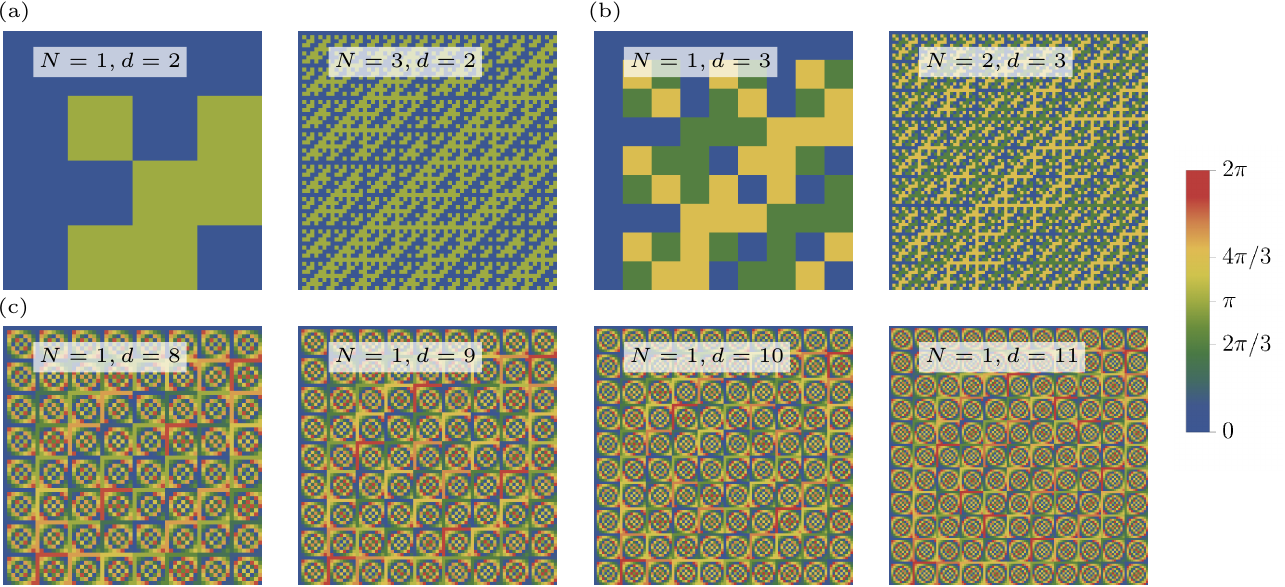}
\caption{
Visualization of the argument of the matrix elements
$\bigotimes_\alpha \qty[F_\alpha\otimes F^*_\alpha](\vec m,\vec n;\vec r,\vec s)$
for different dimensions and number of particles;
rows and columns are indexed by the double indices $(\vec m,\vec n)$ and 
$(\vec r,\vec s)$, respectively.
This matrix map the $\tau(\vec m,\vec n)$ of a Weyl map to the eigenvalues 
$\lambda(\vec r,\vec s)$ of its Choi-\jami{} matrix, see eqs~(\ref{eq:general12}) and 
(\ref{eq:general12.1}). 
We show plots for systems of (a) 
qubits, (b) qutrits, and (c) single-qudits.
Notice that not only the total dimension is relevant, but also the number of 
particles; for instance, compare $N=3$, $d=2$ with $N=1$, $d=8$. 
}
\label{fig:generalized:a:v2}
\end{figure*}

The $\lambda(\vec r,\vec s)$ are the eigenvalues of $\mcD$. They can also be used to characterize 
the Weyl channel $\mcE$. Inverting the relation (\ref{eq:general12}) we get
\begin{subequations}
\begin{eqnarray}
\coef{\vec m}{\vec n}&=&d^{-N}\sum_{\vec r,\vec s}\lambda(\vec r,\vec s)\omega^{-\vec m\cdot\vec r+
\vec n\cdot\vec s},
\label{eq:general13a}\\
\sum_{\vec r,\vec s}\lambda(\vec r,\vec s)&=&d^N.
\label{eq:general13b}
\end{eqnarray}
\label{eq:general13}
\end{subequations}
Here (\ref{eq:general13b}) follows from $\tr\mcD=d^N$, which is a consequence of 
(\ref{eq:general7b}) and (\ref{eq:general8}). 

From (\ref{eq:general8}) and (\ref{eq:general10v2}) follows that the $\coef{\vec m}{\vec n}$ and the $\lambda(\vec r,\vec s)$ are connected by the 
following {\em linear\/} relationship: 
\begin{equation}
\coef{\vec m}{\vec n} =
\sum_{\vec r,\vec s}
\bigotimes_{\alpha}
\qty[ F_\alpha\otimes F_\alpha^\star ]
(\vec m,\vec n;\vec r,\vec s)\lambda(\vec r, \vec s),
\label{eq:general12.1}
\end{equation}
where $F_\alpha$ is the quantum Fourier transform matrix for dimension $d_\alpha$ in the general case, 
and of dimension $d$ in the case we shall generally study (see \Fref{fig:generalized:a:v2}).

We have therefore obtained a full characterization of Weyl channels: choosing arbitrary 
$\lambda(\vec r,\vec s)$ that are positive and add up to $d^N$, the $\coef{\vec m}{\vec n}$ 
given by (\ref{eq:general13a}) define a Weyl channel. 


It is important to highlight that the set of channels introduced in the present work are different from other kinds of generalization of Pauli channels introduced previously \cite{ruskai,Ohno2009,Chruscinski2016}. On the other hand, similar expressions for the eigenvalues associated to random unitary channels on single $d$-level systems have been presented in \cite{Chruscinski2013,Chruscinski2015}. 

\section{Set of extreme points}
\label{sec:extreme}
%
%
%
%
%
%
%
The set of Weyl channels is clearly convex, since equations \eqref{eq:general13} imply that
any Weyl channel is given by a 
convex sum of channels of the form
\begin{equation}
\tau_{\vec{r}_0,\vec{s}_0}(\vec{m},\vec{n}) = \omega^{-\vec{m} \cdot \vec{r}_0 + \vec{n} \cdot \vec{s}_0},
\label{eq:extreme4}
\end{equation}
where $\vec{r}_0, \vec{s}_0$ are fixed vectors whose elements are integer numbers modulo $d$.

Furthermore, we can see that the set of 
Weyl channels is in fact a $d^{2N}-1$ dimensional simplex. 
Recall that all eigenvalues $\lambda(\vec{r},\vec{s})$ of the Choi--\jami{} matrix of a Weyl channel must be non-negative, and sum up to $d^N$ [see
\eref{eq:general13b}].
The set $\lambda(\vec{r},\vec{s})$ is thus the standard $d^{2N}-1$ dimensional simplex. Since the connection \eqref{eq:general12}
between the $\lambda$'s and the $\tau$'s is linear and invertible, then the set of all $\tau$'s is also a 
$d^{2N}-1$ dimensional simplex.  Note however that the $\tau$'s are complex,
so they are actually part of a bigger $2d^{2N}$ dimensional real vector space.
Nonetheless, conditions \eqref{eq:general7} (which are automatically 
satisfied by the formulae  \eqref{eq:general12}) 
additionally limit the $\tau$'s so that the 
number of degrees of freedom is back to $d^{2N}-1$.

Moreover, the extreme points of
the simplex of Weyl channels are given by the $d^{2N}$
channels of equation \eqref{eq:extreme4}.
This is because the extreme points of the $\lambda$'s simplex are
clearly 
\begin{equation}
\lambda(\vec{r},\vec{s}) = d^N \delta_{\vec{r},\vec{r}_0} \delta_{\vec{s},\vec{s}_0}.
\label{eq:extreme5}
\end{equation}
Therefore, those of the set of Weyl channels
are given by applying the transformation (\ref{eq:general13a})
to these extreme points, obtaining as a result the channels of equation \eqref{eq:extreme4}. 
In fact, these channels are the only Weyl channels with the property 
that for all $\vec{m},\vec{n}$, $\abs{\tau(\vec{m},\vec{n})}=1$, as shown in
the following theorem.
\begin{theorem}
\label{thm1}
A Weyl channel is an extreme point of the set of Weyl channels if and only if
$|\tau(\vec{m},\vec{n})|=1$ for all $\vec{m},\vec{n}$.
\end{theorem}
\begin{proof}
We have already proved that extreme points 
are of the form \eqref{eq:extreme4}, and therefore satisfy that
$|\tau(\vec{m},\vec{n})|=1$, so we only need to prove the converse. 
Equation \eqref{eq:general13a} says that
\begin{equation}
\tau(\vec{m},\vec{n}) = d^{-N} \sum_{\vec{r},\vec{s}} \lambda(\vec{r},\vec{s}) \omega^{-\vec m\cdot\vec r+\vec n\cdot \vec s}.
\label{eq:extreme6}
\end{equation}
Recall that $d^{-N}\lambda(\vec{r},\vec{s})$ are non-negative and add up to $1$. It follows from the triangle inequality that if the sum in the right-hand side of \eqref{eq:extreme6} has more than one term, then $\abs{\tau(\vec{m},\vec{n})}<1$.
Thus, a Weyl channel with $\abs{\tau(\vec m,\vec n)}=1$ must have all $\lambda(\vec r,\vec s)$ equal to $0$ except one, say $\lambda(\vec r_0,\vec s_0)$.
In other words, if the Choi--\jami{} matrix of a Weyl channel has only one
eigenvalue $\lambda(\vec r_0,\vec s_0)$ different from zero, then
that Weyl channel is an extreme point.
\end{proof}

The simplest case that illustrates the result of this theorem are the single-qubit 
Weyl quantum channels. Given that the Weyl operators for $d=2$ reduce to 
the Pauli operators, the extremal points for these are the vertices of the 
well-known tetrahedron of qubit quantum channels~\cite{bengtsson_zyczkowski_2017}, 
which we illustrate in Fig.~\ref{fig:simpleces}(a).

We can now characterize in greater detail these extreme points. 
For a given value of $\vec{r}_0, \vec{s}_0$, 
the effect of the channel given by \eqref{eq:extreme4} on a Weyl matrix $U(\vec{m},\vec{n})$ is
\begin{eqnarray}
\mcE_{\vec r_0,\vec s_0}\left[
\unit{\vec m}{\vec n}
\right]&=&\omega^{-\vec r_0\cdot\vec m+\vec s_0\cdot\vec n}
\unit{\vec m}{\vec n}
\nonumber\\
&=&\unit{\vec s_0}{\vec r_0}\unit{\vec m}{\vec n}\unit{\vec s_0}{\vec r_0}^\dagger.
\label{eq:extreme3.1}
\end{eqnarray}
We see therefore that the extreme points of the set of all Weyl channels are unitary channels. Since all Weyl channels are convex combinations of the extreme points \cite{dahl2010introduction}, 
it immediately follows that all the Weyl
channels are simply random unitary channels, constructed from the Weyl unitaries. 

\section{A mathematical structure within Weyl  channels} \label{sec:border} 
In this section, we focus on a subset of Weyl channels with physical relevance
and mathematical beauty. We consider the Weyl channels, which, when
iterated infinitely, converge to channels that completely erase, preserve, or
introduce phases to the projections of the density matrix onto the Weyl
operator basis. Our main results include the characterization of the group
property of this particular subset and a method to determine these channels.
Specifically, we show that the corresponding channels can be obtained by
identifying all subgroups of $\mathbb{Z}_d\oplus \mathbb{Z}_d$ and their
homomorphisms to $\mathbb Z_d$.

\subsection{Subgroup property of Weyl channels} 
\begin{theorem}
\label{thm2}
Let $\coef{\vec m}{\vec n}$ and $\coef{\vec m^\prime}{\vec n^\prime}$ have both norm 1. Then so does 
$\coef{\vec m+\vec m^\prime}{\vec n+\vec n^\prime}$ and additionally
\begin{equation}
\coef{\vec m+\vec m^\prime}{\vec n+\vec n^\prime}=\coef{\vec m}{\vec n}\coef{\vec m^\prime}
{\vec n^\prime}
\label{eq:border1}
\end{equation}
\end{theorem}

\begin{proof}
From (\ref{eq:general13}) follows that, quite generally, $\coef{\vec m}{\vec n}$ are convex combinations
of complex numbers of the form $\omega^k$, with $k$ an integer. Nonetheless, the only such convex combinations
having norm 1 are themselves numbers of the form $\omega^k$, with $k$ an integer, 
therefore $\tau(\vec m,\vec n)$ and $\tau(\vec m',\vec n')$ are, under the hypotheses of 
the theorem, of the form $\omega^k$ and $\omega^{k'}$, respectively.
Similarly for $\coef{\vec m^\prime}{\vec n^\prime}$, which we take to be equal to $\omega^{k^\prime}$. 

Setting $\tau(\vec{m},\vec{n}) = \omega^k$ in \ref{eq:general13a} and conveniently
 rewriting the equation results in
\begin{equation}\label{eq:tomas}
\omega^k=d^{-N}\sum_l\omega^l\sum_{\vec r,\vec s}\lambda(\vec r,\vec s)\delta_{-\vec m\cdot\vec r+\vec n\cdot\vec s,l}
\end{equation}

and a similar expression for $\omega^{k'}$, replacing
$\vec m$ and $\vec n$ with their primed versions. 
Since $\lambda(\vec r,\vec s)$ are positive and sum up to $d^N$, it follows that the right-hand side is a convex sum of $\omega^l$. For it to equal $\omega^k$, an extreme point, only the term with $l=k$ can be different from zero. It follows that 
$\lambda(\vec r,\vec s)=0$,
whenever $-\vec m\cdot\vec r+\vec n\cdot\vec s\neq k$ or $-\vec m^\prime
\cdot\vec r+\vec n^\prime\cdot\vec s\neq k^\prime$.
From this follows straightforwardly that
again $\lambda(\vec r,\vec s)=0$ whenever
\begin{equation}
-(\vec m+\vec m^\prime)\cdot\vec r+(\vec n+\vec n^\prime)\cdot\vec s\neq (k+k^\prime)
\label{eq:border3}
\end{equation}
which implies 
\begin{equation}
\coef{\vec m+\vec m^\prime}{\vec n+\vec n^\prime}=\omega^{k+k^\prime}=
\coef{\vec m}{\vec n}\coef{\vec m^\prime}
{\vec n^\prime}
.
\label{eq:border4}
\end{equation}
\end{proof}


\begin{figure*} 
\centering
\hspace*{.5mm}
\includegraphics[page=1,height=.4\textwidth]{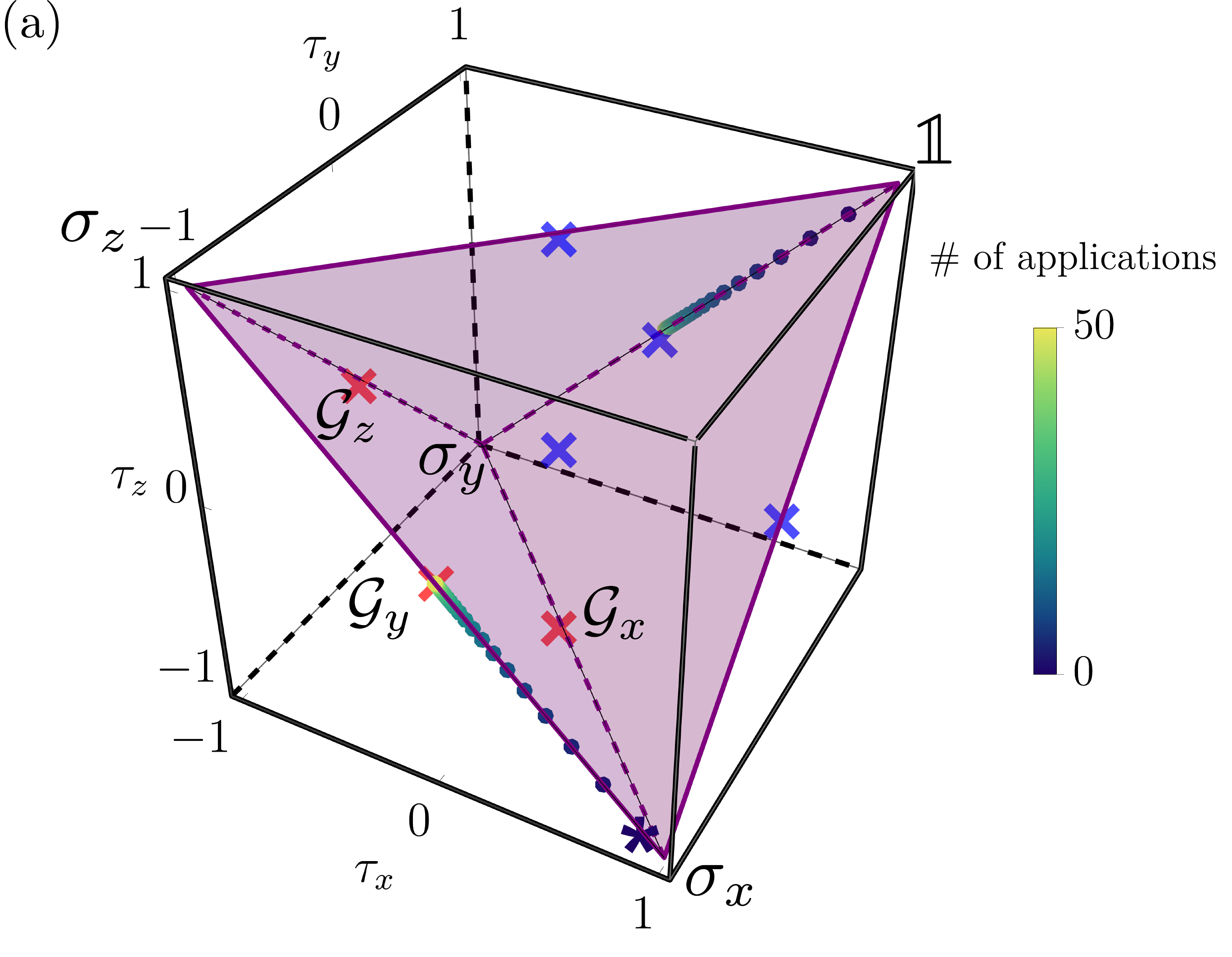}
\includegraphics[height=.4\textwidth]{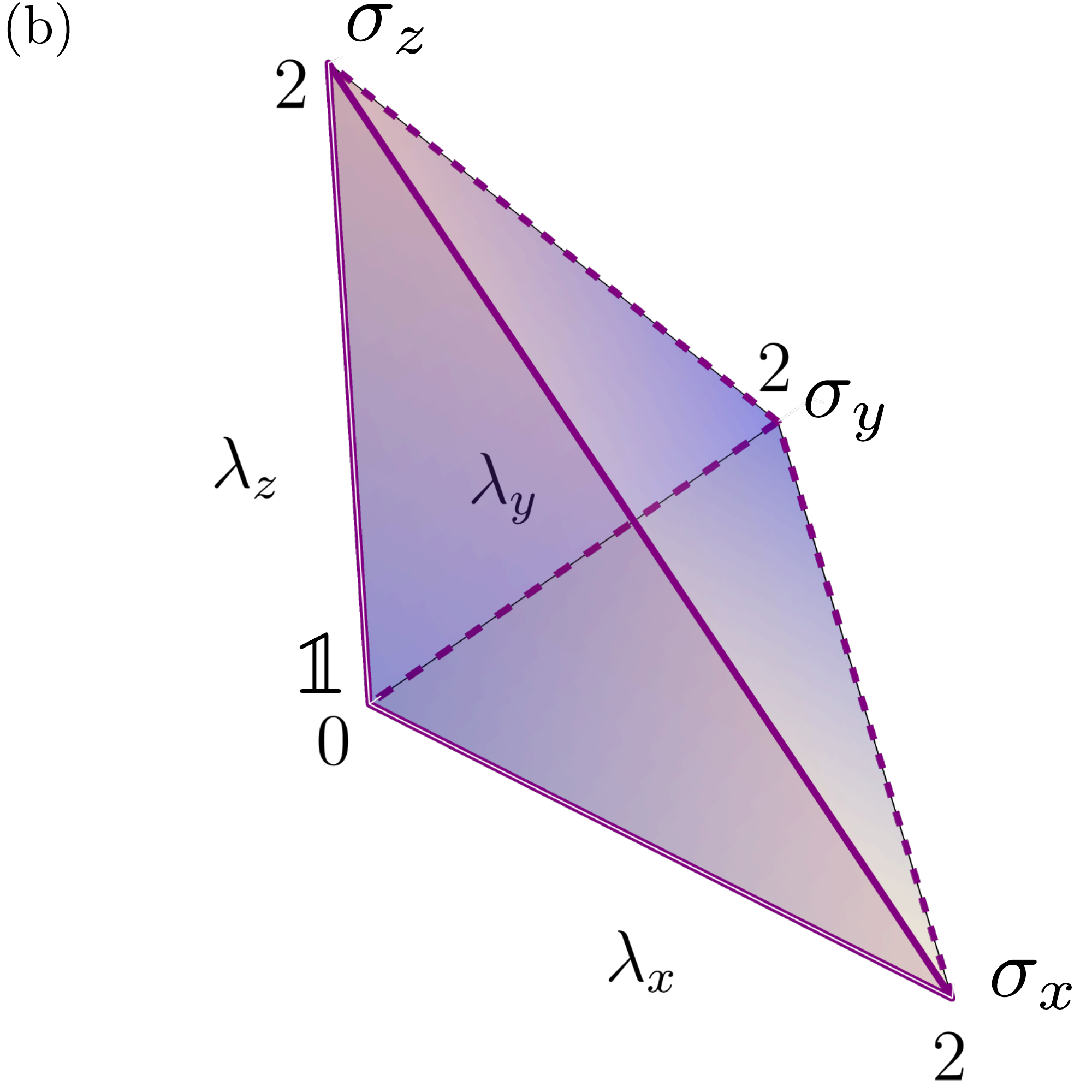}\hspace*{\fill}

\includegraphics[height=2.18cm]{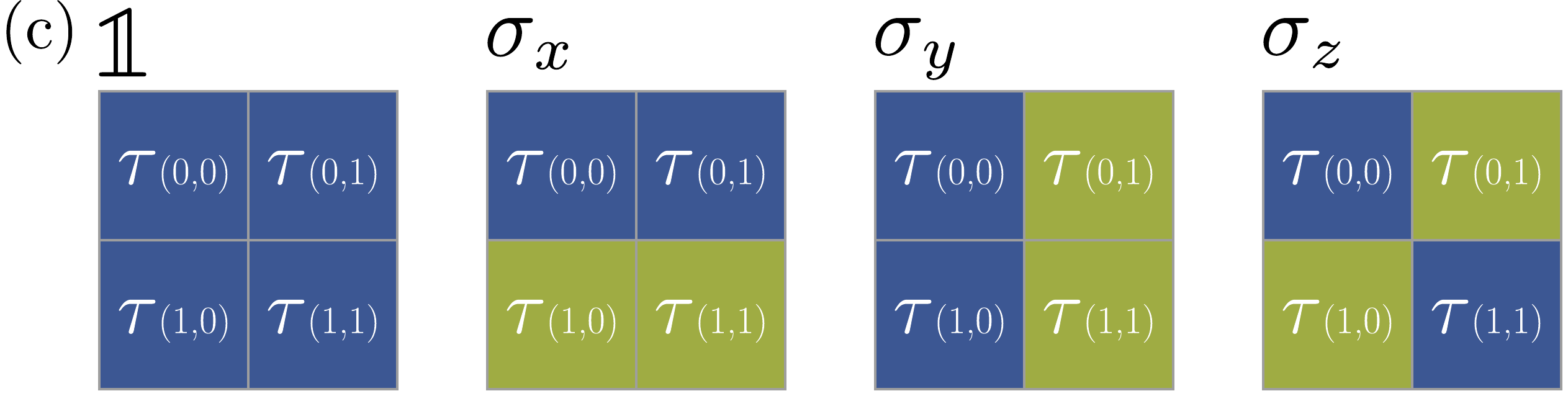}\hspace*{6.2mm}
\includegraphics[height=2.18cm]{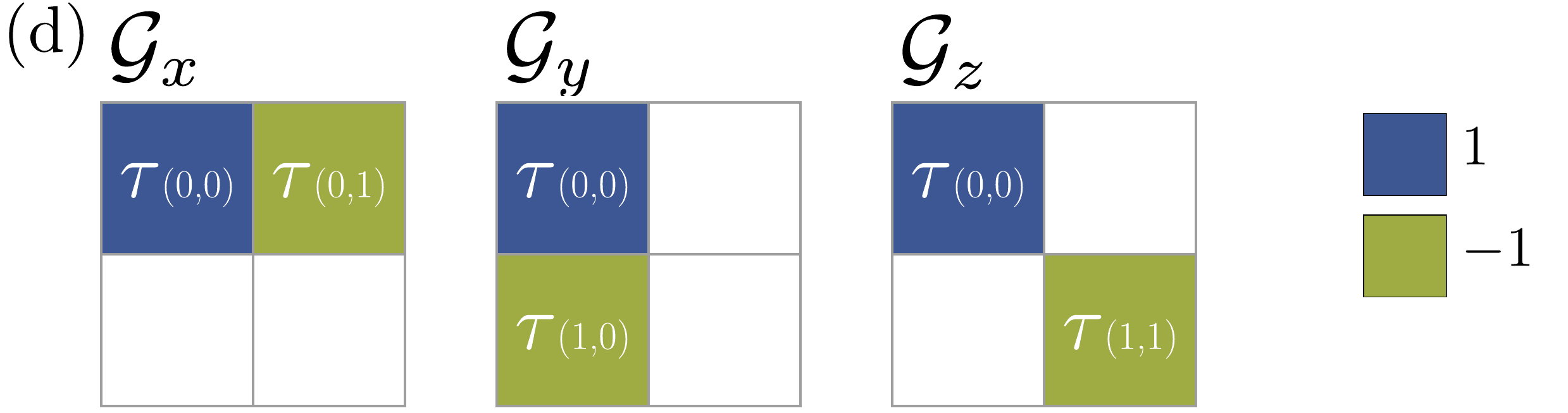}
\caption{
Simplexes of (a) the $\tau(m, n)$ of all single-qubit Weyl channels, in which 
we identify $\tau(0,1) = \tau_x$, $\tau(1,0) = \tau_y$, and $\tau(1,1) = \tau_z$, as usual, and similarly for $\lambda$s; and of (b) the 
eigenvalues $\lambda(m, n)$ of the corresponding Choi-\jami{} matrix. Additionally, we 
depict in (c) and (d) the extreme points of (a) and Weyl erasing generators, respectively.
}
\label{fig:simpleces}
\end{figure*} 

This means that the set of all $(\vec m,\vec n)$, such that $\coef{\vec m}{\vec n}$, has norm 1
form an {\em additive subgroup\/} of the abelian group
\begin{equation}
\mcG=\mathbb{Z}_{d}^{\oplus N}\oplus \mathbb{Z}_{d}^{\oplus N}
\label{eq.border5}
\end{equation}
with respect to vector addition 
modulo $d$. Note that this is one case where a significant difference
arises when the $N$ particles
have different dimensions $d_\alpha$; the vectors $(\vec m, \vec n)$ would then belong to the group
\begin{align}
\mcG=\left(\mathop{\bigoplus}_{\alpha=1}^N\mathbb{Z}_{d_\alpha}\right)\oplus\left(\mathop{\bigoplus}_{\alpha=1}^N
\mathbb{Z}_{d_\alpha}\right).
\label{eq:fdsfads}
\end{align}
In other words, the set $\mcH\subseteq\mcG$ on which $\tau(\vec m,\vec n)$ has norm 1 forms a 
subgroup of $\mcG$ and $\tau$ can be seen as a homomorphism from $\mcH$ to $\mathbb{Z}_d$.

To determine all $\tau(\vec m, \vec n)$ of Weyl channels satisfying eq. \eqref{eq:border1}, we must proceed in 2 steps: 
\begin{enumerate}
\item Determine all subgroups $\mcH\subseteq\mcG$.
\item Determine all homomorphisms from $\mcH$ to $\mathbb{Z}_d$.
\end{enumerate}

In the following section we present an algorithm to determine 
all subgroups $\mcH$ of $\mcG$ and homomorphisms from $\mcH$ to $\bbZ_d$.

We wish to remark that given a quantum channel for which some
of its coefficients $\tau(\vec m,\vec n)$ satisfy Eq. \ref{eq:border1}, the rest
of coefficients are not necessarily null. However, still these are restricted by the
complete positivity condition, Eq. \ref{eq:general12}.

\subsection{Weyl channels } \label{sec:subgroups:and:homomorphisms} 
To determine all $\tau(\vec{m},\vec{n})$ of Weyl channels
satisfying eq. \eqref{eq:border1} we proceed in two steps.
The first step involves identifying all subgroups of $\mathcal{G}$ 
[cf. \eref{eq:fdsfads}]. We begin by stating two relevant facts 
about finite abelian groups, and then discuss how to find the subgroups
for the more general case of an abelian group $\mathcal{G}$, 
which encompasses the majority of our discussion in this section. 
After that, we describe the second step, which
is how to determine all phases of $\coef{\vec m}{\vec n}$ of a WCE channel
by determining all homomorphisms from a subgroup to the roots of unity.
We refer the reader to Appendix \ref{app:examples} to illustrate with 
several examples the algorithm we present in this section.

Whenever $p$ and $q$ are coprime, the group $\bbZ_{pq}$ is isomorphic 
to $\bbZ_p \oplus \bbZ_q$. Therefore, we may
use the prime decomposition  of $d_{\alpha}$ to separate each $\bbZ_{d_{\alpha}}$ in  eq.
\eqref{eq:fdsfads} as a sum of cyclic groups of prime power order.
We proceed in this way for all $\alpha$ in eq.
(\ref{eq:fdsfads}), and then we group the terms corresponding to different primes, so $\mcG$ can be written as
\begin{equation}\label{eq:border6.2}
\mcG = \bigoplus_p\mcG_p,
\end{equation}
with 
$\mcG_p=\bigoplus_i\bbZ_{p^{k_i}}$,
for each prime $p$ that appears in the decomposition
of any of the $d_{\alpha}$.

Since the direct sum of two arbitrary abelian groups of orders $m$ and $n$
that are coprime yield all abelian groups of order $mn$,
we can directly construct all subgroups of $\mcG$ by finding
all the subgroups of each $\mcG_p$. 
In other words, although $\mcG$ in \eref{eq:fdsfads} may have a complicated 
decomposition, we focus only in determining the subgroups of $\mcG_p$,
which will be convenient to write as
\begin{equation}\label{eq:Gp}
\mcG_p = \bigoplus_{\alpha=1}^r \bbZ_{p^{M_{\alpha}} }
\end{equation}
where $M_{\alpha}$ are in non-increasing order.

The group $\mcG_p$ is associated with the sequence 
$\mpart = M_{1} \,\ldots\,M_{r}$, 
which is a \textit{partition} of $M=\sum_{\alpha}M_{\alpha}$.
Therefore, we will refer to $\mcG_p$ as a group of type $\mpart$.
Furthermore, for any partition of $M$ there exists an abelian group of order $p^M$ that is unique
up to an isomorphism \cite{Burnside1911}.
If one has a subgroup $\mcH_p$ of $\mcG_p$, the corresponding partition, let us 
call it $\overline N$, satisfies $N_\alpha \leq M_\alpha$. On the other hand, 
once the choice of the non-increasing order for the partition of the group 
$\mcG_p$, the corresponding partitions for the subgroups $\mcH_p$ inherit a
well-defined order from the group, and the corresponding partitions cannot 
therefore be taken in non-increasing order.

Another important fact about finite abelian groups is that they all 
have a basis; that is, they can be generated by the
integer combinations of a set of elements. In our particular case, 
a simple way of choosing a basis is by picking a generating element for 
each cyclic group in \eref{eq:Gp}. We denote them $\vec e_{\alpha}$,
and therefore an arbitrary $h\in\mathcal H$ can be {\em uniquely\/} expressed as
\begin{equation}
h=\sum_{\alpha=1}^r n_{\alpha}\vec e_{\alpha},
\label{eq:border7}
\end{equation}
where $n_{\alpha}\in\mathbb{Z}_{p^{M_{\alpha} }}$, 
and the multiplication of a group element
by an integer $m$ is defined as the addition of the group element to itself repeated $m$ times. The 
number $r$ of elements in the basis is independent of the choice of basis and 
it is known as the group's
rank $r$.

The general idea for finding all subgroups of $\mcG_p$ is 
to determine a subset of subgroups such that, upon applying all automorphisms 
$T: \mcG_p \mapsto \mcG_p$, all  others are found. 
We will say that two subgroups of $\mcG_p$ are $T$-isomorphic when there is 
an automorphism $T$ mapping one to the other. 
Then, to find the subgroups of $\mcG_p$ we first 
determine any subset with the maximum number of subgroups that are not $T$-isomorphic.
We call these  ``representative subgroups''.
By definition, applying all automorphisms $T$ (which we describe how to find in Appendix \ref{app:b}) to the representative subgroups all other subgroups of $\mcG_p$ are found.

Note the difference between the concept of isomorphism for the subgroups, and the concept 
of $T$-isomorphism. The latter depends not only on the group structure of the subgroup 
${\cal H}$, but also on the way in which it is embedded in the group ${\cal G}_p$. For 
instance, we can embed the group $\bbZ_2$ in the 
group $\bbZ_{2^2}\oplus \bbZ_2$ either as a 
subgroup of the first summand or as a subgroup of the second. In other words, the 
partition $\overline M$ describing the full group is $\overline M = 2\,1$
and the subgroup $\bbZ_2$ can be embedded with a partition  
$0\,1$ as well as $1\,0$.
The two subgroups, being both isomorphic to $\bbZ_2$, are abstractly isomorphic, but that isomorphism cannot be extended to an isomorphism of $\bbZ_4\oplus \bbZ_2$. 

All subgroups of $\mathcal G_p$ are found applying all its
automorphisms $T$ to the subgroups generated by the bases 
\begin{equation}\label{eq:border:B:1}
\mathcal B=\{ p^{s_1}\vec e_1, \ldots, p^{s_r}\vec e_r\}\quad 0 \leq s_{\alpha} \leq M_{\alpha}.
\end{equation}
Nevertheless, more than one different selection $\mathbb S = \{s_\alpha\}$
may determine two bases of subgroups $T$-isomorphic, in other words, that are connected by an automorphism of $\mathcal G_p$. For example, consider a group $\mathcal G_p$ of 
type $\mpart=2\,2\,1\,1$. The partitions $\mathbb S=0\,1\,0\,1$ and $\mathbb S'=1\,0\,1\,1$ determine bases of $T$-isomorphic subgroups, because 
the automorphism defined as $T(\vec{e}_1) = \vec{e}_2$, 
$T(\vec{e}_2) = \vec{e}_1$, 
$T(\vec{e}_3) = \vec{e}_4$ and $T(\vec{e}_4) = \vec{e}_3$ maps one 
to the other. 
From each of these  $T$-isomorphic sets of subgroups  we can pick an
arbitrary element, which will be called \textit{the representative subgroup}. 

To find the representative subgroups we need a criterion to determine when two
bases $\mathcal B$ of the form  \eqref{eq:border:B:1} generate
$T$-isomorphic groups. Let us denote $\tilde M_1,\ldots,\tilde M_q$ the $q$
different values in the sequence of numbers in $\mpart$ (for instance, if
$\bar{M} = 2211$, then $q=2$ and $\tilde{M}_1 = 2, \tilde{M}_2 =1$).
Furthermore, we define the subset $S_j=\{s_\alpha, \, \forall\, \alpha
:M_\alpha=\tilde M_j\}$ of $\mathbb S$, that is, $S_j$ is the subset of
$\mathbb{S}$ formed by all the  $s_\alpha$ whose $\alpha$s correspond to the
indices of the $M_\alpha$ that are equal to $\tilde{M}_j$. Then, the criterion
is the following: two different sets $\mathbb S$ and $\mathbb S'$ determine
bases of $T$-isomorphic subgroups whenever their corresponding subsets $S_j$
and $S_j'$ are the same for all $j$.

We are ready to describe the complete algorithm to determine 
all the subgroups of a given group $\mcG$. 
First, decompose $\mathcal G$ as a sum of prime power order groups 
$\mathcal G_p$. For every $\mathcal G_p$, find all sets 
$\mathbb{S} = \{s_{\alpha}\}$ and discriminate between them to find the only
ones that determine representative subgroups.
Then apply to them all automorphisms $T$ of $\mathcal G_p$, so
all subgroups of $\mathcal G_p$ will be found, albeit with repetitions. 
A description of the 
group of automorphisms of an arbitrary abelian group $\mcG$ is provided in \cite{Hillar2007}, 
and the technique is summarized for completeness' sake in Appendix \ref{app:b}.
Finally, to find the subgroups of $\mathcal G$ apply the direct sum 
between all different subgroups of each $\mathcal G_p$.

Furthermore, a way to count the total number of subgroups of $\mcG_p$
is already known in the literature. 
Since any abelian group of prime power order can only have subgroups that are also of prime
power order, subgroups of order $p^L$, with $L<M$,  can also be characterized by a partition
$\lpart$  of $L$.
It is shown in \cite{Burnside1911, Birkhoff1935} that necessary and sufficient conditions for the partition $\lpart$ to correspond
to a possible subgroup of the group determined by the partition 
$\mpart$ of $M$ are
\begin{subequations}
\begin{align}
L_{\alpha}&=0\quad (\alpha >r),
\label{eq:border9a}\\
L_{\alpha}&\leq M_{\alpha},
\label{eq:border9b}\\
L_{\alpha}&\geq L_{\alpha+1}.
\label{eq:border9c}
\end{align}
\label{eq:border9}
\end{subequations}
An expression for the number of different subgroups of type $\lpart$ is 
already known in the literature. For that matter, we refer the reader to Appendix 
\ref{app:NumberOfSubgroups}.

To fully determine the coefficients $\coef{\vec m}{\vec n}$ with norm 1, 
we interpret them as a function 
%
%
that maps $\mcH$ to the group of roots of unity $\omega^j$. We consider 
$\coef{\vec m}{\vec n}=\omega^{\phi(\vec m,\vec n)}$,
thus we are looking for all homomorphisms 
$\phi:\bigoplus_{M_\alpha}\bbZ_{p^{M_\alpha}} \mapsto \bbZ_{p^{M_1}}$.
To determine one of such functions uniquely, it is sufficient
to specify the values of $\phi$ on a basis of $\mcH$, as described in Appendix
\ref{app:c}. 




Note that all the above remarks greatly simplify when $d$ is a prime number. In that case 
the group $\mcG$ is additionally a vector space. The set of subgroups
can then be described as the set of vector subspaces using the usual techniques of linear algebra. 
All the partitions described above then reduce to partitions of the type where $M_\alpha$ is either 
1 or 0, and the partition is fully characterized by the number of its non-zero elements, which correspond
to the subspace's {\em dimension}. Finally, the homomorphism $\tau$ can be described as a 
linear map from the vector space $\mcG$ to the field $\mathbb{Z}_d$, which is once more 
straightforwardly  described in terms of linear algebra. 
\section{Weyl erasing channels} \label{sec:WCE} 
\subsection{Generalities} 
\label{subsec:gener1}

%

In this section we focus on a particular class of Weyl channels; those for which 
 $\abs{\tau(\vec m',\vec n')}=0$ or $1$. 
In other words, we will discuss Weyl channels that completely erase, preserve or introduce 
specific phases to the projections of the density matrix of a system of qudits onto the 
Weyl matrices basis. We will refer to this subset of Weyl channels as Weyl erasing channels.

Weyl erasing channels are an interesting subset of Weyl channels, as they arise
from the composition of one or more of these channels an infinite number of times.
For instance, the infinite composition of any Weyl channel
for which all $\abs{\tau(\vec m,\vec n)}<1$, except $\tau(\vec 0,\vec 0)=1$,
results in the completely depolarizing channel.
In the general case, the repeated application of a Weyl channel 
may not converge to a single Weyl erasing 
channel, however it oscillates between two or more such channels. 
For example, applying many times the single-qubit Weyl
channel depicted with an asterisk in Fig. \ref{fig:simpleces}(a) asymptotically oscillates 
between the channel collapsing the Bloch sphere onto the $y$ axis and other channel 
collapsing it onto the $y$ axis while reflecting it across
the $x$-$z$ plane. This oscillation continues indefinitely between $\vec\tau=(1,0,1,0)$
and $\vec\tau'=(1,0,-1,0)$.


In the following, we derive a Kraus representation of Weyl erasing channels
that will provide insight into the physical implementation of these channels.
We then focus on deriving an expression of the eigenvalues of the Choi-\jami{}
matrix of Weyl erasing channels exclusively, as we already have an expression for all Weyl
channels [\textit{c.f.}~\eqref{eq:general12}].  
For this, we begin presenting an algorithm that uses the mathematical machinery developed 
in section \ref{sec:subgroups:and:homomorphisms} to find all $\tau(\vec m,\vec n)$ of 
any Weyl erasing channel: 
\begin{enumerate}
	\item Find all sets of indices $\{(\vec m,\vec n):\abs{\tau(\vec m,\vec n)}=1\}$ by 
	determining all subgroups $\mcH \subset \mcG$, with $\mcG$ that in 
	\eqref{eq:border6.2}.
	\item Find the values of $\tau(\vec m,\vec n)=\omega^{\phi(\vec m,\vec n)}$
	for all $(\vec m,\vec n)\in\mcH$ by determining all homomorphisms $\phi:\mcH\mapsto
	\bigoplus_p\bbZ_{p^{M_{1}(p) }}$, with $M_1(p)$ such that $\bbZ_{p^{M_{1}(p) }}$
	denotes the largest order cyclic group for every $\mcG_p$ in \eqref{eq:border6.2}.
	\item Assign $\tau(\vec m',\vec n')=0$ for all $(\vec m',\vec n')\notin\mcH$.
\end{enumerate}
While an exhaustive enumeration of all Weyl erasing channels for a large number $N$ of
particles is not practical, the construction provides insights into the mathematical
structure of this set. 
In fact, we show examples of Weyl erasing channels for a
single-qubit, a 4-level system and a system composed by both of them in
Figs.~\ref{fig:z2z2:subgroups}-\ref{fig:WCEPs}.
We remark that the algorithms to determine $\tau(\vec m,\vec n)$ of Weyl 
and Weyl erasing channels are both the same up until determining all homomorphisms.
Then, for the former one may assign any value to the $\tau(\vec m',\vec n')\notin\mcH$
as long as they 
keep the channel completely positive, whereas for the latter one must assign the value zero
to all $\tau(\vec m',\vec n')\notin\mcH$.

Let us now evaluate the eigenvalues of the Choi matrix of a Weyl erasing
channel. To be consistent with the previous section, we  consider Weyl erasing
channels of a system of $N$ particles, each with 
dimension a power of $p$, so the group in question is $\mcG_p$ [see eq. \eqref{eq:Gp}]. 
The group's rank is then $r=2N$.
Recall these channels have
$\coef{\vec m}{\vec n}=0$ for all $(\vec m,\vec n)\notin\mcH$, thus, we find from (\ref{eq:extreme6}) 
\begin{equation}
\lambda(\vec r,\vec s)= d^{-N}
\sum_{(\vec m,\vec n)\in\mcH}\coef{\vec m}{\vec n}
\prod_{\alpha=1}^{N}
\omega_{\alpha}^{m_\alpha r_\alpha - n_\alpha s_\alpha}
\label{eq:WCE1}
\end{equation}
where $\omega_{\alpha}=\exp\qty(2\pi i/p^{M_{2\alpha} })$ 
Since the composition of two Weyl channels is simply to 
multiply both sets of $\tau(\vec m,\vec n)$, we can see that $\tau(\vec m,\vec n)$ 
of Weyl erasing channels are those of an extreme Weyl channel 
[\textit{c.f.}~\eqref{eq:extreme4}] for all $(\vec m,\vec n)\in\mcH$, and
$\tau(\vec m',\vec n')=0$ for all $(\vec m',\vec n')\notin\mcH$. Hence, substituting 
$\tau(\vec m,\vec n)$, and considering $\omega_{ p^{ \alpha }}^k = 
\omega_{p^{\beta}}^{p^{(\alpha - \beta)k}}$, $\alpha > \beta$, we can write
\begin{align}
	\lambda(\vec r,\vec s)&=
	\sum_{(\vec m,\vec n)\in\mcH}
	\omega_{p^{M_{1} }}^{f(\vec m,\vec n)},
	\label{eq:WCE5:2}
\end{align}
where $f(\vec m,\vec n)=\sum_{\alpha=1}^N 
p^{M_1 - M_{2\alpha - 1}}(m_\alpha(r_\alpha - r_{0,\alpha}) 
- n_\alpha(s_\alpha-s_{0,\alpha}))$. 
To evaluate this expression we note that $f$ is an homomorphism 
$f:\mcG_p\mapsto \bbZ_{p^{M_2}}$. Therefore, the sum evaluates to zero unless $f$ 
maps $\mcH$ to the trivial group:
\begin{align}
	\lambda(\vec r,\vec s)&=\begin{cases*}
		\left|
		\mcH
		\right|&\qquad $f(\vec m,\vec n)=0$ 
		for all $(\vec{m},\vec{n})\in\mcH$\\
		0&\qquad otherwise,
	\end{cases*}
	\label{eq:WCE9}
\end{align}
where $|\mcH|$ is defined as the number of elements of $\mcH$. 
Furthermore, let us consider the group
$\mcH^\perp=\{(\vec m,\vec n):f(\vec m,\vec n)=0\}$.
The only non-zero $\lambda(\vec r,\vec s)$ are those with indices 
$(\vec r,\vec s)$ such that  $(\vec r - \vec r_0,\vec s - \vec s_0) \in \mcH^\perp$.

Finally, having obtained the eigenvalues of the Choi-\jami{} 
matrix of Weyl erasing channels we describe their canonical Kraus representation.
Recall that that Weyl matrices are the Kraus operators of Weyl channels  with probabilities
equal to $\lambda(\vec{r}, \vec{s})$ [\textit{c.f.} \eref{eq:extreme3.1}]. 
It then follows that Kraus operators of Weyl erasing channels are the subset of
Weyl matrices $U(\vec{r}, \vec{s})$, 
with $(\vec r - \vec r_0,\vec s - \vec s_0) \in \mcH^\perp$,
each with probability $|\mcH^\perp|/|\mcG|$. 

\subsection{Generators} 
\label{subsec:gener2}


In the following, we investigate the smallest subset of Weyl erasing channels 
which, under composition, generate the whole set. For the sake of 
simplicity, we start by finding the generators of Weyl channels with 
$\tau(\vec m,\vec n)=0$ or $1$, as these are Weyl channels characterized 
only by subgroups of $\mcG$. Subsequently, we move to the most general 
Weyl erasing channels that either preserve, erase or introduce phases to 
the density matrix.

We shall determine those subgroups that are Indecomposable, in the sense that they
cannot be generated as the non-trivial composition of two Weyl channels. We call these the generator subgroups.
We consider once again the group $\mcG$ shown in 
eq. \eqref{eq:border6.2}, thus, we first discuss how to determine the generator 
subgroups of $\mcG_p$, and, from those, determine the generator subgroups of $\mcG$.
Similarly to what we did in section \ref{sec:subgroups:and:homomorphisms},
the generator subgroups $V_p$ of $\mathcal{G}_p$ can be found 
constructing representative generator subgroups $V_p^\star$ 
and applying all automorphisms of $\mcG_p$ to them.

We claim that a representative subgroup is a generator $V_p$ of $\mcG_p$
if and only if its basis is of the form
\begin{align}\label{base-gen:2}
\mathcal{B}_{V_p} = 
\{\vec{e}_1, \ldots, \vec{e}_{j-1}, p^{s_j} \vec{e}_j, \vec{e}_{j+1}, 
  \ldots, \vec{e}_r\},\, 1\leq s_j \leq M_j.
\end{align}
This is verified as follows.
Consider a subgroup $\mathcal{H}_p$ with basis \eqref{eq:border:B:1}
such that its set of values $\mathbb{S}=\{s_{\alpha}\}_\alpha$
has two (or more) values $s_{\beta},s_{\gamma}\ne 0$. That is, 
$\mcH_p$ has the basis 
$\mathcal{B}_{\mcH_p} = \{\vec{e}_1, \ldots, \vec{e}_{\beta - 1}, 
p^{s_\beta} \vec{e}_\beta, \ldots, \vec{e}_{\gamma - 1}, 
p^{s_\gamma} \vec{e}_\gamma, \ldots, \vec{e}_r\}$.
Then, $\mathcal{H}_p$ can be expressed as the non-trivial intersection of the group 
$g_p'$ with basis 
$\mathcal{B}_{V'_{p}} = \{\vec{e}_1, \ldots, \vec{e}_{\beta - 1}, 
p^{s_\beta} \vec{e}_\beta, \ldots, \vec{e}_r\}$
and another group $V_{p}''$ with basis 
$\mathcal{B}_{V_{p}''} = \{\vec{e}_1, \ldots, \vec{e}_{\gamma - 1}, 
p^{s_\gamma} \vec{e}_\gamma, \ldots, \vec{e}_r\}$.
Now, we check that a group $\mathcal{H}_p$ with a 
basis of the form \eqref{base-gen:2}
cannot be expressed as an intersection of two subgroups containing $\mathcal{H}_p$
\textit{strictly}. Note that groups $\mcH_p^\prime$ satisfying $\mcH_p\subsetneq\mcH_p^\prime\subset\mcG_p$
must have a basis of the form
$\{\vec{e}_1, \cdots , p^{s_j'} \vec{e}_j, \cdots, \vec{e}_r \}$,
with $s_j' < s_j$. Therefore, if we have two such groups $\mathcal{H'}_p$ and
$\mathcal{H''}_p$, then the group arising from intersection $\mcH_p' \cap \mcH_p''$
has a basis  $\{\vec{e}_1, \cdots , p^{\max(s_j', s_j'')} \vec{e}_j, \cdots, \vec{e}_r \}$,
which doesn't generate $\mathcal{H}_p$ because the integer span of 
$p^{s_{j'}}\vec e_{j'}$ or $p^{s_{j''}}\vec e_{j''}$ are strictly larger than the integer span 
of $p^{s_{j}}\vec e_{j}$.

Finally, to find all generator subgroups $V$ of $\mathcal{G}=\bigoplus_{p} \mathcal{G}_{p}$
we proceed as follows. We begin by finding the generator subgroups of $\mcG_p$.
Then, the generator subgroups $V$ are the subgroups of the form 
\begin{align}\label{eq:subgroups}
	V = \bigoplus_p H_p,
\end{align}
where $H_p = \mcG_p$ for all $p$ except for one $p=p'$, for which
$H_{p'}=V_{p'}$.
To see this, notice that  $V$ is a generator since
any other subgroup $\mathcal{H'}$ strictly containing it
must also have $H'_p = \mathcal{G}_{p}$ for $p \neq p'$ 
and $H'_{p'} \subsetneq V_{p'}$.
Therefore, for two such subgroups $\mathcal{H'}$ and $\mathcal{H''}$ to generate $V$,
$H'_{p'}$ and $H''_{p'}$ should generate $V_{p'}$, which is impossible since
$V_{p'}$ is a generator of $\mathcal{G}_{p'}$.
On the other hand, any subgroup of $\mcG$ that is composed as the sum of
groups $\mcG_p$ except for two (or more) primes $p'$ and $p''$
in which the sum is of any corresponding generator subgroup, 
may be generated by two subgroups $\mcH$ of the form 
\eqref{eq:subgroups} with suitable generators $H_{p'}=V_{p'}$ and 
$H_{p''}=V_{p''}$.

To summarize up to this point: to determine the generators of Weyl erasing channels with 
$\tau(\vec m,\vec n) =0,1$ one is only required to determine the generator 
subgroups of the corresponding group in which the indices $(\vec m,\vec n)$ of
$\tau(\vec m,\vec n)=1$ live. To determine the generators of Weyl erasing channels 
that also introduce phases one must still find generator subgroups, and a suitable
homomorphism for each of them as well.

More specifically, a Weyl erasing generator with $\abs{\tau(\vec m,\vec n)}=0,1$
is completely characterized by a generator subgroup of $\mcG$ and a
homomorphism $\phi:\mcG \mapsto \bigoplus \bbZ_{p^{M_{2}(p) }}$, 
with $\bbZ_{p^{M_{2}(p) }}$ is the largest order cyclic group for every $\mcG_p$.
The composition of two channels with subgroups $\mcH_1$ and $\mcH_2$ and
homomorphisms $\phi_1$ and $\phi_2$ yields a channel corresponding to
$\mcH_1\cup\mcH_2$ and homomorphism $\phi_1+\phi_2$. Therefore, for every
generator subgroup $V$, if we define a basis $\phi_{V}^\alpha$ of the space of
homomorphisms, then the set of channels corresponding
to each generator subgroup $V$, together with any element 
$\phi_{V}^\alpha$ mapping $V$ to the whole group $\bbZ_{p^{M_{2}(p) }}$,
form a set of generators. For instance, we show a set of Weyl erasing generators 
of a $4$-level system in Fig. \ref{fig:q4bit:wce:generators}

To avoid misunderstandings, note that the ``generators'' for channels with 
$\tau(\vec m,\vec n)=0,1$, those characterized completely by a generator subgroup 
(and one may say that also by a homomorphism $\phi=0$),
are not generators of the whole set of Weyl erasing channels 
($\abs{\tau(\vec m,\vec n)}=0,1$). The former channels can always be obtained 
through iterated composition of channels with a non-zero homomorphism.
For example, the single-qubit Weyl erasing channels depicted with a blue point in Fig.
\ref{fig:simpleces}(a), except for the one located at $\vec\tau=(0,0,0)$, are Weyl 
generator channels with $\tau(m,n)=0,1$. However, they can be obtained from composing
two times each of the three Weyl generator channels depicted with a red point.
Those are the generators of all Weyl erasing channels of a single-qubit. 

\begin{figure}
\centering
\includegraphics[width=\columnwidth]{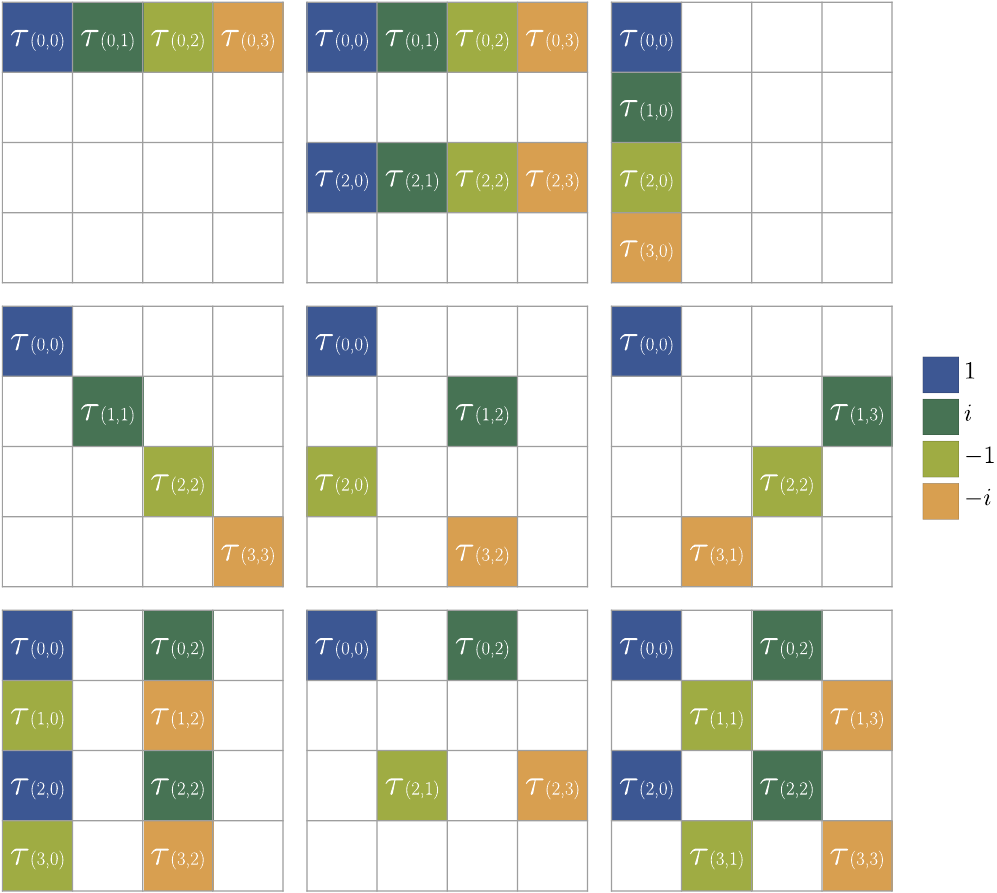}
\caption{Weyl erasing generators for a single-particle of dimension $d=4$. 
Concatenating in all possible ways the corresponding channels one obtains all the 
Weyl quantum channels of a 4-level system with $\abs{\tau(m,n)}=0,1$.
}
\label{fig:q4bit:wce:generators}
\end{figure}

\section{Conclusions}\label{sec:conclusions} 

In this paper, we explored a class of quantum channels of multipartite systems with
different-dimensional particles. Our focus was on extending the study 
of Pauli and Weyl channels; the former have been studied for 
systems of many qubits and the latter for single $d$-level systems.

We begin by introducing the multi-particle Weyl operators $U(\vec m,\vec n)$
to define a Weyl map as a diagonal map in this basis; its eigenvalues are denoted
by $\tau(\vec m,\vec n)$.
We derived the constraints on 
$\tau(\vec m,\vec n)$ for a Weyl map to preserve the trace and 
hermiticy of the density matrix, as well as to be completely positive.
For the latter, we diagonalized its Choi-\jami{} matrix and found the 
linear relationship between these eigenvalues and $\tau(\vec m,\vec n)$.

Several features of Weyl channels emerged from our study. We
identified the extreme points of the set, 
which correspond to Weyl operators, highlighting 
the random unitary nature of Weyl channels.
Additionally, we established a subgroup structure within a Weyl channel,
showing that the indices ${(\vec m,\vec n)}$ of all $\abs{\tau_{\vec m,\vec n}} = 1$ 
is a subgroup of the direct product of groups 
$\bbZ_{d_\alpha} \oplus \bbZ_{d_\alpha}$, where $d_\alpha$ represents the 
dimension of the $\alpha$-th particle.

Furthermore, we introduced Weyl erasing channels, which are Weyl channels that
either preserve, erase, or introduce phases to the Weyl operator basis. 
These extend the concept of \textit{component erasing} channels.
Given that all channels of this type exhibit the subgroup structure, 
with the remaining $\tau_{\vec m,\vec n}$ set to zero,
we were able to find the smallest subset which generate, under composition, 
the whole set.

Our work contributes to the understanding of many-body quantum channels, and,
moreover, to the dynamics of $d$-level systems, which have growing importance 
for both theoretical and practical purposes. Future research directions include
investigating divisibility, non-Markovianity, channel capacity, and the subset of
entanglement-breaking channels among other properties of the Weyl channel set.
\begin{acknowledgments} 

Support by projects CONACyT 285754, 254515 and UNAM-PAPIIT IG100518, IG101421 is
acknowledged.  
J.~A.~d.~L. acknowledges a scholarship from CONACyT.
J.~A.~d.~L. would like to thank Cristian Álvarez for valuable dicussions 
about finite groups. A. F.
acknowledges funding by Fundação de Amparo à Ciência e Tecnologia do Estado de
Pernambuco - FACEPE, through processes BFP-0168-1.05/19 and BFP-0115-1.05/21.
D. D. acknowledges OPTIQUTE APVV-18-0518, DESCOM VEGA-2/0183/21 and Štefan
Schwarz Support Fund.
\end{acknowledgments} 
\appendix
\section{Computation of the Choi--\jami{} matrix} \label{app:choi} 

This appendix demonstrates that the Choi-\jami{} matrix of any diagonal map $\mcE$ takes the form
\begin{align}\label{eq:choi_general}
\mcD = 
\frac{1}{d^N} \sum_{\vec m,\vec n} \tauv \Uv \otimes \Uv^\star,
\end{align}
whenever ${U(\vec m,\vec n)}$ form an orthogonal basis of unitaries~\cite{ruskai}.

The Choi--\jami{} matrix of a quantum map $\mcE$ is defined as follows:
\begin{equation}\label{eq:app1.1}
\mcD=
\frac{1}{d^N}
\sum_{\vec k,\vec l}\mcE\left[
\ket{\vec k}\bra{\vec l}
\right]\otimes
\left(
\ket{\vec k}\bra{\vec l}
\right).
\end{equation}
We may now express $\ket{\vec k}\bra{\vec l}$ in terms of any orthogonal basis 
of unitaries $U(\vm, \vn)$
\begin{align}\label{eq:app1.2}
\dyad{\vec k}{\vec l} = 
\frac{1}{d^N}
\sum_{\vec m,\vec n} 
\Tr\qty(U^\dagger(\vm, \vn)\dyad{\vec k}{\vec l}) 
U(\vm, \vn).
\end{align}
Substituting this in the expression (\ref{eq:app1.1}) for $\mcD$ yields.
\begin{widetext}
\begin{subequations}
\begin{align}	
\frac{1}{d^N}
\sum_{\vec k,\vec l} \mcE\qty(\dyad{\vec k}{\vec l}) \otimes \dyad{\vec k}{\vec l}
&=
\frac{1}{d^{2N}}
\sum_{\vec k,\vec l,\vm,\vm',\vn,\vn'} 
\mcE\qty[ \Tr\qty(U^\dagger(\vm, \vn)\dyad{\vec k}{\vec l})U(\vm, \vn) ] 
\otimes \qty[ \Tr\qty(U^\dagger(\vm', \vn')\dyad{\vec k}{\vec l})U(\vm', \vn') ] \\
&= 
\frac{1}{d^{2N}}
\sum_{\vec k,\vec l,\vm,\vm',\vn,\vn'} 
\Tr\qty(U^\dagger(\vm, \vn)\dyad{\vec k}{\vec l})
\Tr\qty(U(\vm', \vn')\dyad{\vec l}{\vec k})
\tau(\vm, \vn) U(\vm, \vn) \otimes U(\vm', \vn')^* ,\\
\intertext{where we have used the complex conjugate of \eqref{eq:app1.2}. Then,
using the definition of trace it follows}
&= 
\frac{1}{d^{2N}}
\sum_{\vec k,\vec l,\vm,\vm',\vn,\vn'} 
\matrixel{l}{U^\dagger(\vm, \vn)}{k} \matrixel{k}{U(\vm', \vn')}{l}    
\tau(\vm, \vn) U(\vm, \vn) \otimes U(\vm', \vn')^* \\
&= 
\frac{1}{d^N}
\sum_{\vm,\vn} \tau(\vm, \vn) U(\vm, \vn) \otimes U(\vm, \vn)^* .
\end{align}
\end{subequations}
\end{widetext}
This expression can also be obtained also using a recently result of Siewert, in which 
he derives an expression for the maximally entangled state in terms of an 
arbitrary orthogonal basis~\cite{siewert2022orthogonal}.
\section{Eigenvalues of $U(m,n)$ and of $U(m,n) \otimes U(m,n)^*$}
\label{app:eigenvalues}
We will find the eigenvalues of $U(m,n)$.
Since it is unitary, we will express the eigenvalues as $\omega^c$, with $c\in\mathbb{R}$.
Let us consider an eigenvector $|\phi\rangle = \sum_r \phi(r) |r\rangle$ with eigenvalue $\xi = \omega^c$.
The eigenvalue equation for $U(m,n)$ leads to the following relation:
\begin{equation}
\label{eq: recurrence}
\phi(r+n) = \omega^{-mr} \omega^c \phi(r).
\end{equation}
Starting with an arbitrary index $r$ and applying this recursion equation $l-1$ times, we obtain
\begin{equation}
\phi(r+nl) = \omega^{-lmr - \frac{1}{2} l(l-1)mn + cl} \phi(r).
\end{equation}
In the particular case in which $l = l' := \dfrac{d}{\mathrm{gcd}(d,n)}$ we may 
use that $l'n$ is a multiple of $d$, so:
\begin{equation}
\phi(r) = \omega^{-l'mr - \frac{1}{2} l'(l'-1)mn + cl'} \phi(r),
\end{equation}
which implies that (for values of $r$ such that $\phi(r) \neq 0$): 
\begin{equation}
-l'mr - \frac{1}{2} l'(l'-1)mn + cl' = sd
\end{equation}
for some integer $s$. Therefore:
\begin{equation}
c = \dfrac{sd}{l'} + \dfrac{1}{2}(l'-1)mn + mr = \mathrm{gcd}(d,n) s + mr + \dfrac{1}{2}(l'-1)mn.
\end{equation}
So we conclude that all eigenvalues of $U(m,n)$ necessarily have the form 
\begin{equation}
\omega^{\mathrm{gcd}(d,n) s + mr + \frac{1}{2} (l'-1)mn}
\end{equation}
for $s$ and $r$ integers. 

Furthermore, taken modulo $d$, the set $\{\mathrm{gcd}(d,n) s + mr \; |\; s,r \in \mathbb{Z}_d\}$ is equivalent to $\{n s + mr \; |\; s,r \in \mathbb{Z}_d\}$ which is also equivalent to $\{\mathrm{gcd}(m,n) k \; |\; k \in \mathbb{Z}_d\}$. Therefore, the $d$ eigenvalues of $U(m,n)$ have are
\begin{equation}
\label{eq: eigenvalues-U}
\xi = \omega^{\mathrm{gcd}(m,n) k + \frac{1}{2}(l'-1)mn}.
\end{equation}
From this, it is straightforward that the eigenvalues of 
%
%
$U(m,n) \otimes U(m,n)^*$ are
\begin{equation}
\mu(r,s) = \omega^{\mathrm{gcd}(m,n) k - \mathrm{gcd}(m,n) h}.
\end{equation}
This set is equivalent to
\begin{equation}
\label{eq: eigenvalues_UxU}
\mu(r,s) = \omega^{mr - ns}
\end{equation}
where $r,s$ are integers modulo $d$.
\section{Automorphisms of finite abelian groups}
\label{app:b}

In the following we describe  the bijective homomorphisms $T$ of an arbitrary abelian group. 
Without loss of generality we limit ourselves to groups that are the direct sum of groups
of the type $\mathbb{Z}_{p^M}$, specifically
\begin{equation}
\mcG=\bigoplus_{\alpha=1}^r\mathbb{Z}_{p^{M_\alpha}}.
\label{eq:b1}
\end{equation}
To fix notations, we shall work with a fixed basis $\vec e_\alpha$, $1\leq\alpha\leq r$, where $r$ is the
{\em rank\/} of $\mcG$. The map $T$ is therefore uniquely determined by the values of 
$T\vec e_\alpha$. Since $\vec e_\alpha$ is a basis, we can write
\begin{equation}
T\vec e_\alpha=\sum_{\beta=1}^rt_{\alpha\beta}\vec e_\beta.
\label{eq:b2}
\end{equation}
The $t_{\alpha\beta}$ are then uniquely determined, if we view them as homomorphisms from 
$\mathbb{Z}_{p^{M_\beta}}$ to $\mathbb{Z}_{p^{M_\alpha}}$.  Since such homomorphisms can always 
be expressed through the multiplication by some appropriate number, the expression 
given in (\ref{eq:b2}) is meaningful. 

Now let us specify more precisely the range of variation of the $t_{\alpha\beta}$. We distinguish two cases
\begin{enumerate}
\item $M_\alpha\leq M_\beta$: in this case any number modulo $p^{M_\alpha}$ will do, and two different
such numbers provide different homomorphisms.

\item $M_\alpha>M_\beta$: in this case, the number needs to be a multiple of $p^{M_\alpha-M_\beta}$, since 
otherwise it is not possible to define the map. In that case, we may describe $t_{\alpha\beta}$ as
$p^{M_\alpha-M_\beta}\tau_{\alpha\beta}$, where $\tau_{\alpha\beta}$ is an arbitrary number modulo
$p^{M_\beta}$
\end{enumerate}

 Consider the matrix $T$ in greater detail, and just as in the main text,
let us denote by $\tilde{M}_1, \cdots, \tilde{M}_q$ the
{\em distinct\/} values of $M_\alpha$
in {\em strictly decreasing order}. 
We define $\nu_{\alpha}$ to be the number of
times $\tilde{M}_\alpha$ 
appears repeated in the original series. 
This defines a division of the $T$ matrix in {\em blocks\/} 
of size $\nu_\alpha\times\nu_\beta$, where $1\leq\alpha,\beta\leq q$.

We first take the elements $t_{\alpha\beta}$ modulo $p$. As a consequence of the observation (2), 
all blocks with $\alpha<\beta$ are filled with zeros, whereas all other blocks have arbitrary entries.
It thus follows that the matrix is invertible modulo $p$ if and only if all the diagonal blocks are invertible. 
The number of invertible $\nu_\alpha\times\nu_\alpha$ matrices modulo $p$ is given by
\begin{equation}
I_\alpha=\prod_{\beta=1}^{\nu_\alpha}\left(
p^{\nu_\alpha}-p^{\beta-1}
\right).
\label{eq:b3}
\end{equation}
One sees this by observing that we may first choose an arbitrary non-zero vector of length 
$\nu_\alpha$ in $p^{\nu_\alpha}-1$ different ways, then chose a second vector independent 
from the first, and so on. 

All the other entries in the blocks below the diagonal, that is, the $t_{\alpha\beta}$
with $\alpha>\beta$, can be chosen arbitrarily.
If we thus define
\begin{equation}
K_0=\sum_{1\leq\beta<\alpha\leq q}\nu_\alpha\nu_\beta,
\label{eq:b4}
\end{equation}
then the total number of possible forms of the matrix $T$ modulo $p$ is
\begin{equation}
N(p)=p^{K_0}\prod_{\alpha=1}^qI_\alpha.
\label{eq:b5}
\end{equation}

We now need to work out the number of ways this can be extended to the full matrix, where the entries 
have the full range of variation specified above. Note first that the condition of invertibility carries over 
automatically upon extension, as the inverse matrix of $T$ modulo $p$ can be extended uniquely to the
inverse of the extended matrix. 

To the entries on or below the diagonal, that is, with
$t_{\alpha\beta}$ such that $\alpha\geq\beta$, we can add any number of the form 
$p\tau_{\alpha\beta}$, where $\tau_{\alpha\beta}$
is an arbitrary number taken modulo $p^{\tilde{M}_\alpha-1}$. So the number of possibilities of extending these
blocks is given by $p^{K_1}$, where
\begin{equation}
K_1=\sum_{1\leq\beta\leq\alpha\leq q}(\tilde{M}_\alpha-1)\nu_\alpha\nu_\beta.
\label{eq:b5.1}
\end{equation}
For the blocks above the diagonal, that is, the blocks with $t_{\alpha\beta}$ such that $\alpha<\beta$, 
they are of the form $p^{\tilde{M}_\alpha-\tilde{M}_\beta}\tau_{\alpha\beta}$, with $\tau_{\alpha\beta}$ a number 
modulo $p^{\tilde{M}_\beta}$, so that the total number of ways of extending the blocks above the diagonal
is $p^{K_2}$, with 
\begin{equation}
K_2=\sum_{1\leq\alpha<\beta\leq q}\tilde{M_\beta}\nu_\alpha\nu_\beta.
\label{eq:b6}
\end{equation}
The final result for the total number of automorphisms is thus given by
\begin{equation}
N_{tot}(M_1,\ldots,M_r)=p^{K_0+K_1+K_2}\prod_{\alpha=1}^qI_\alpha.
\label{eq:b7}
\end{equation}

\section{Homomorphisms from $\mcH$ to the cyclic group $\mathbb{Z}_d$}
\label{app:c}

Here we describe the set of homomorphisms $\phi$
from an abelian group of the form
$\mcH = \bigoplus_{\alpha=1}^r \mathbb{Z}_{p^{M_{\alpha}}}$ to the cyclic group $\mathbb{Z}_{p^{M_1}}$.
As always, the numbers $M_{\alpha}$ are ordered in decreasing order.

We may as always choose a basis $\vec e_\alpha$ of $\mcH$, each having order
$p^{M_\alpha}$. The homomorphism $\phi$ is then uniquely determined by a set of homomorphisms
$\phi_\alpha$ from the cyclic groups $\mathbb{Z}_{p^{M_\alpha}}$ to $\mathbb{Z}_{p^{M_1}}$. 

Whenever $M_1 = M_\alpha$, $\phi_\alpha$ simply reduces to multiplication by an arbitrary 
$r_\alpha$ number modulo 
$p^{M_1}$. On the other hand, if $M_\alpha< M_1$, then $\phi_\alpha$ is given by the multiplication by 
a number of the form $p^{M_1-M_\alpha}r_\alpha$ where $r_\alpha$ is an arbitrary number modulo
$p^{M_\alpha}$. 

If we therefore define $\nu$ as the number of $M_\alpha = M_1$, so that $M_\nu\geq M_1$ but 
$M_{\nu+1}< M_1$, and $\nu=0$ if $M_\alpha<M_1$ for all $\alpha$, then $\phi$ can be expressed 
as follows: 
\begin{subequations}
\begin{eqnarray}
\phi\left(\sum_{\alpha}c_\alpha\vec e_\alpha
\right)&=&\vec\phi\cdot\vec c:=\sum_\alpha\phi_\alpha c_\alpha
\label{eq:c1a}\\
\phi_\alpha&=&\begin{cases*}
p^{M_1-M_\alpha}s_\alpha&$(\alpha\geq\nu)$\\
t_\alpha&$(\alpha<\nu)$
\end{cases*}
\label{eq:c1b}
\end{eqnarray}
\label{eq:c1}
\end{subequations}
where $s_\alpha$ and $t_\alpha$ are numbers modulo $p^{M_\alpha}$ and $p^{M_1}$ respectively. 

The total number of such homomorphisms is therefore given
by $p^K$ with
\begin{equation}
K=\sum_{\alpha=1}^\nu M_\alpha+M_1(r-\nu)
\label{eq:c2}
\end{equation}

\section{Examples}
\label{app:examples}
To illustrate the application of the mathematical tools presented in the main text, 
we will provide detailed examples of how to identify 
the Weyl channels as is described
in section \ref{sec:border}.
Remember that said channels are characterized by a subgroup $\mathcal{H}$
of the group of indices $\mathcal{G}$ (which corresponds to the 
indices whose $\tau$'s have norm $1$)
and an homomorphism (which gives the phases to each $\tau$).
The examples will be arranged in increasing generality, starting with a system
of one qudit with prime dimension and ending with the most general
case of many qudits of arbitrary dimensions.

\subsection{Single particle with prime dimension}
\label{subsec:single-prime-particle}
Here we show how to follow the algorithm described in section \ref{sec:border} 
for the case of one qudit with prime dimension $d=p$. 
In this case, the group of indices for the $\tau$'s is simply $\mathbb{Z}_p \oplus \mathbb{Z}_p$
and we search for all its subgroups and then all homomorphisms to $\mathbb{Z}_p$.
While we do this in general, we simultaneously show 
the specific case for $p=2$ (a single qubit).\\

We start determining the types of subgroups of $\mathbb{Z}_p \oplus \mathbb{Z}_p$ 
and a representative subgroup of each type. 
For that, we take the following steps:
\begin{itemize} 
\item We select a basis for $\mathbb{Z}_p \oplus \mathbb{Z}_p$, the simplest would be  $\{\vec{e}_1, \vec{e}_2\} := \{(1,0),(0,1)\}$.
\item Define $M_{\alpha}$ as the number such that $p^{M_{\alpha}}$ 
is the order of $\vec{e}_{\alpha}$. In this case, $M_1 = M_2 = 1$ 
and therefore the partition for the group is $\bar{M} = M_1M_2 = 11$.
\item Find all the sets $\mathbb{S} = \{s_\alpha\}$ with 
$0 \leq s_{\alpha} \leq M_{\alpha}$.
In this case, there are four said sets: $
\{0,0\}, \{0,1\}, \{1,0\}, \{1,1\}$ 
\item For each set, we define the basis
 $\mathcal{B} = \{p^{s_1} \vec{e}_1, p^{s_2} \vec{e}_2\}$, and therefore get the following bases:
\begin{widetext}
\begin{align}
\label{eq:bases-qubit}
\mathbb{S} = \{0,0\} \rightarrow \mathcal{B}=\{p^0 \vec{e}_1 , p^0 \vec{e}_2\} = \{\vec{e}_1, \vec{e}_2\}\;,\; \mathbb{S} = \{0,1\} \rightarrow \mathcal{B}= \{p^0 \vec{e}_1 , p^1 \vec{e}_2\} = \{\vec{e}_1\}, \\
\mathbb{S} = \{1,0\} \rightarrow \mathcal{B}=\{p^1 \vec{e}_1 , p^0 \vec{e}_2\} = \{\vec{e}_2\}\;,\;
\mathbb{S} = \{0,0\} \rightarrow \mathcal{B}=\{p^0 \vec{e}_1 , p^0 \vec{e}_2\} = \{ \}.
\end{align}
\end{widetext}
Notice that $p \vec{e}_{\alpha} = (0,0)$,
which doesn't contribute to the basis.

\item Now we only keep bases that are not T-isomorphic, so as to avoid unnecessary redundancies when applying automorphisms in the next step.

As mentioned in the main text, we do so by first defining the sequence of numbers
$\tilde{M}_1, \cdots, \tilde{M}_q$
given by the $q$ different values in
the sequence of numbers in $\bar{M}$.
In this case, as $\bar{M} = 11$, there is only one number in said sequence, being $\tilde{M}_1 = 1$. 
Then, we define the subsets $S_j = \{s_{\alpha}, \forall \alpha : M_{\alpha} = \tilde{M}_j\}$. 
In this case, we only have one such subset, which happens to be the whole set, $S_1 = \{s_1, s_2\}$.
As described in the main text, two different 
bases of the ones described in 
the previous step 
are T-isomorphic if all their sets $S_j$ are the same.
In this case, this means that the bases constructed 
from $\mathbb{S} = \{1,0\}$ and $\mathbb{S} = \{0,1\}$
are T-isomorphic.
Therefore, we may only keep one of those bases, let's say we keep 
$\{0,1\}$ and discard the other one,
so that the representative bases are:

\begin{equation}
\label{eq:bases-qubit2}
\{\vec{e}_1, \vec{e}_2\}, \;\{\vec{e}_1\},\; \{ \}.
\end{equation}

\end{itemize}

The subgroups generated by these bases are the
``representative subgroups''. 
Then, to find all possible subgroups,
we need to find all automorphisms of $\mathbb{Z}_p \oplus \mathbb{Z}_p$ 
and apply them
to these representative groups, so that we can obtain all subgroups of each type 
starting from the representatives. \\

As shown in eq.(\ref{eq:b2}), automorphisms of
$\mathbb{Z}_p \oplus \mathbb{Z}_p$ 
are determined by a matrix $t_{\alpha \beta}$ with dimensions $r \times r$,
where $r$ is the number of elements in the basis of the group.
Therefore, in this case the automorphisms are characterized by $2\times 2$ matrices,
where each entry $t_{\alpha \beta}$ can be a number modulo $p$.

Furthermore, these entries are constrained
by the conditions given in Appendix \ref{app:b},
and since in this case $M_1 = M_2$, all $t_{\alpha\beta}$
fall into case 1 of the aforementioned conditions. 
This implies that all $t_{\alpha \beta}$ are numbers modulo $p^{M_{\alpha}} = p$.
This gives a total of $p^4$ possible matrices, but 
we need to only keep those that are invertible.
For example, for the especial case of one qubit, we construct all $2 \times 2$
matrices such that all entries $t_{\alpha \beta}$
are numbers modulo $2$,
and out of these $16$ matrices, only $6$ of them are invertible:
\begin{align}
T_1 &= \begin{pmatrix}
0 & 1 \\
1 & 1
\end{pmatrix}, \quad
T_2 &= \begin{pmatrix}
1 & 0 \\
1 & 1
\end{pmatrix}, \quad
T_3 &= \begin{pmatrix}
1 & 1 \\
0 & 1
\end{pmatrix}, \\
T_4 &= \begin{pmatrix}
1 & 1 \\
1 & 0 
\end{pmatrix}, \quad
T_5 &= \begin{pmatrix}
0 & 1 \\
1 & 0 
\end{pmatrix}, \quad
T_6 &= \begin{pmatrix}
1 & 0\\
0 & 1
\end{pmatrix}.
\end{align}
Therefore, these matrices represent the $6$ possible
automorphisms of $\mathbb{Z}_2 \oplus \mathbb{Z}_2$.\\

To find all the subgroups, we simply apply all automorphisms
on each representative subgroup found in eq.(\ref{eq:bases-qubit2}).
For the case of one qubit, the result would be as follows:
\begin{itemize}
\item $\{\vec{e}_1, \vec{e}_2\}$: 
The group generated by this basis is the whole gruop,
and when we apply any automorphism,
we always get back the whole group
because automorphisms are invertible.
Therefore, the only group here is the whole group $\mathbb{Z}_p \oplus \mathbb{Z}_p$.

\item $\{\vec{e}_1\}$: For this basis,
the representative subgroup it generates is $\{(0,0),(1,0)\}$.
As before, we apply all the automorphisms to this subgroup.
We can see an example for the case of one qubit, when applying $T_1$ to the element$(1,0)$, we get as a result $(0,1)$, 
so applying $T_1$ to the subgroup gives as a result $\{(0,0), (0,1)\}$. 
Similarly, when using the other automorphisms in the case of one qubit, we get the following subgroups (excluding repetitions):
\begin{equation}
\label{eq: subgroups qubit 10}
\{(0,0), (1,0)\}, \quad \{(0,0), (0,1)\}, \quad \{(0,0), (1,1)\}.
\end{equation} 
\item $\{\}$: In this case,
the representative subgroup is the trivial $\{(0,0)\}$ and applying any automorphism leaves this subgroup intact.
\end{itemize} 

Therefore, we have found all subgroups of $\mathbb{Z}_p \oplus \mathbb{Z}_p$.
For the case of $p=2$,
they are: the complete group,
the subgroups obtained in eq. (\ref{eq: subgroups qubit 10})
and the trivial subgroup $\{(0,0)\}$, see Fig~\ref{fig:z2z2:subgroups}.
That is, up until this point we have found the sets of indices 
$\{m,n\}$ for which $\tau(m,n)$ can have norm $1$ in a Weyl channel.

\begin{figure}[h!]
\includegraphics[width=\columnwidth]{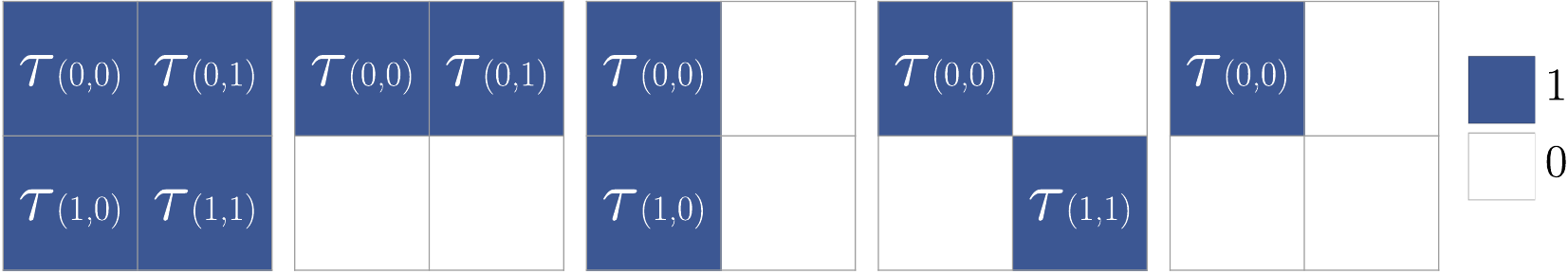} 
\caption{
Single-qubit Weyl erasing channels with $\tau(m,n)=0,1$. These are completely 
characterized by the sets $\{ (m,n) : \tau(m,n)=1 \}$, which 
are the subgroups of $\bbZ_2 \oplus \bbZ_2$.}
\label{fig:z2z2:subgroups}
\end{figure}

Now we find all homomorphisms $\phi: \mathbb{Z}_p \oplus \mathbb{Z}_p \rightarrow \mathbb{Z}_p$. 
A homomorphism is characterized by its value in each element in the basis,
$\phi_{\alpha} :=\phi(\vec{e}_{\alpha})$.
To know the possible values of $\phi_{\alpha}$, we first define $\nu$ like in Appendix \ref{app:c},
as the number such that $M_{\nu} \geq n$ but $M_{\nu+1} < n$ 
(where $n$ is the exponent of the co-domain of $\phi$, in this case the co-domain is $\mathbb{Z}_p$, so that $n=1$
and we can see that $\nu = 2$). 
The possible values of $\phi_{\alpha}$ are given by the cases in eq. (\ref{eq:c1b}):
\begin{itemize}
\item $\phi_1$: Since $\alpha=1 < \nu = 2$, we have the second case and therefore 
$\phi_1$ is a number modulo $p^{n} = p$.
\item $\phi_2$: Since $\alpha = 2 \geq \nu =2$, we have the first case, so that $\phi_2 = p^{n-M_2} s_{s} = s_{2}$ with $s_{2}$ a number modulo $p^{M_2} = p$. Therefore, $\phi_2$ is also a number modulo $p$. 
\end{itemize}
To find all possible homomorphisms, we determine all possible pairs $\phi_1,\phi_2$. 
Since $\phi_1,\phi_2$ can be any number modulo $p$, we
have a total of $p^2$ homomorphisms. In the special case of 
one qubit, they are: $\phi_1 = \phi_2 = 0;\; \phi_1 = 0 ,\phi_2 = 1;\;
\phi_1 = 1 ,\phi_2 = 0;\;$ and 
$\phi_1 = \phi_2 = 1$. We show in Fig. \ref{fig:qubit:weyl:erasing} all Weyl 
erasing channels of a single qubit.

\begin{figure}[h!]
\includegraphics[width=.8\columnwidth]{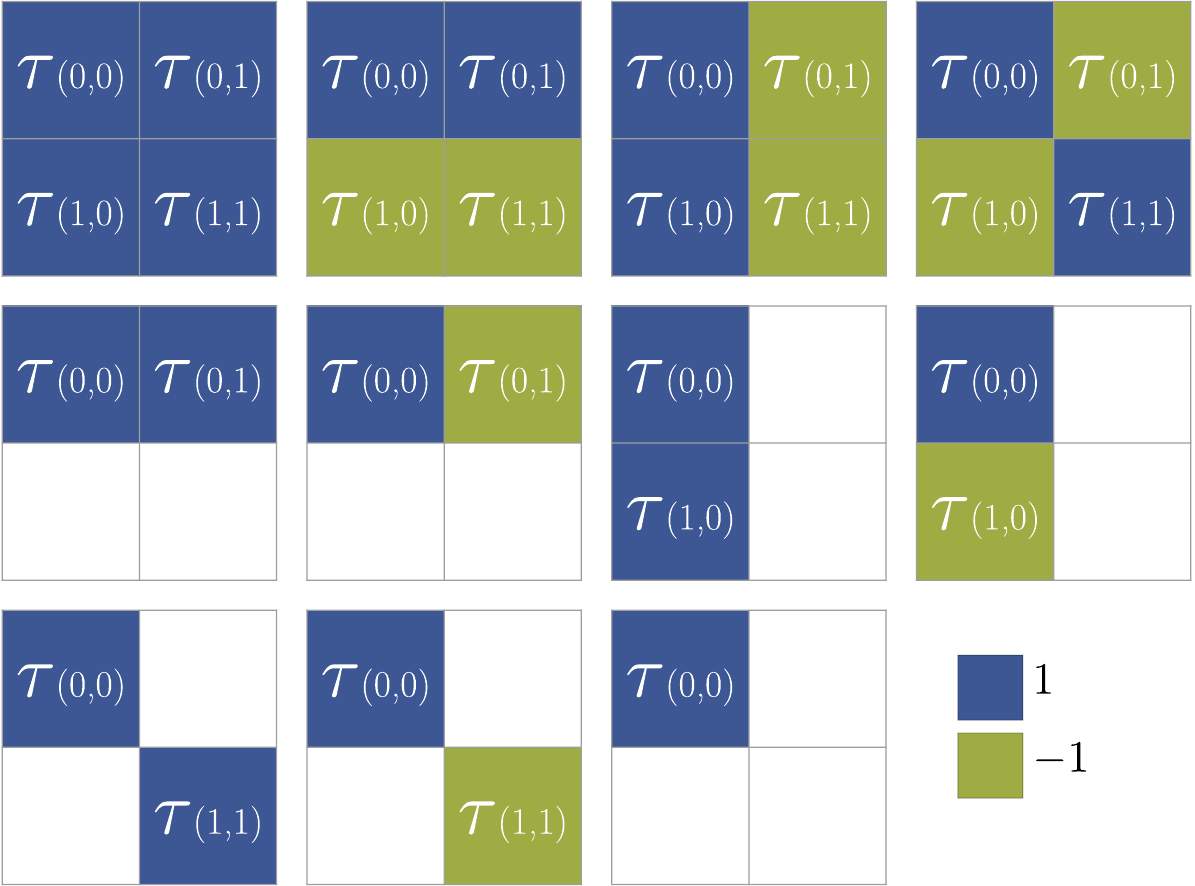} 
\caption{
Single-qubit Weyl erasing channels with $\abs{\tau(m,n)}=0,1$. These are completely 
characterized by two elements: (i) the sets $\{ (m,n) : \abs{\tau(m,n)} = 1 \}$, which 
are the subgroups of $\bbZ_2 \oplus \bbZ_2$, and (ii) all homomorphisms 
$\phi : \bbZ_2 \oplus \bbZ_2 \mapsto \bbZ_2$.
}
\label{fig:qubit:weyl:erasing}
\end{figure}

\subsection{Single particle with $d = p^n$}

Now we generalize to the case of a single particle with $d=p^n$. To have a concrete example
to show, we consider a particle with $d= 2^2 = 4$.\\

\begin{itemize}
\item We select a basis of $\mathbb{Z}_{p^n} \oplus Z_{p^n}$, for example, $\{\vec{e}_1, \vec{e}_2\} = \{(1,0),(0,1)\}$.
\item Define $M_{\alpha}$ as the number such that $p^{M_{\alpha}}$ is the order of $\vec{e}_{\alpha}$. 
In this case, $M_1 = M_2 = n$, so the partition of 
the group is $\bar{M} = M_1M_2 = nn$. 
For the special case of a 4-level system, the partition is $\bar{M} =22$.
\item Find all the sets $\mathbb{S} = \{s_{\alpha}\}$ with $0 \leq s_{\alpha} \leq M_{\alpha}$. 
For the case of a 4-level system, there are nine said sets: $\{0,0\}$, $\{0,1\}$, $\cdots$, 
$\{2,1\},$ $\{2,2\}$.
\item For each set, we define a basis
$\mathcal{B} = \{p^{s_1}\vec{e}_1, p^{s_2} \vec{e}_2\}$. 
For example, for a 4-level system, the bases are:
\begin{widetext}
\begin{eqnarray}
& \mathbb{S} = \{0,0\} \rightarrow \mathcal{B}= \{\vec{e}_1, \vec{e}_2\}, \;
\mathbb{S} = \{0,1\}\rightarrow \mathcal{B}=\{\vec{e}_1, 2 \vec{e}_2\}, \;
\mathbb{S} = \{0,2\}\rightarrow \mathcal{B}= \{\vec{e}_1\}, \\
& \mathbb{S} = \{1,0\} \rightarrow \mathcal{B}= \{2\vec{e}_1, \vec{e}_2\}, \;
\mathbb{S} = \{1,1\}\rightarrow \mathcal{B}=\{2 \vec{e}_1, 2\vec{e}_2\}, \;
\mathbb{S} = \{1,2\}\rightarrow \mathcal{B}= \{2\vec{e}_1\},  \\
& \mathbb{S} = \{2,0\} \rightarrow \mathcal{B}= \{\vec{e}_2\}, \;
\mathbb{S} = \{2,1\}\rightarrow \mathcal{B}=\{2 \vec{e}_2\}, \;
\mathbb{S} = \{2,2\}\rightarrow \mathcal{B}= \{\}. 
\end{eqnarray}
\end{widetext}
\item We define the sequence of numbers $\tilde{M}_1, \cdots, \tilde{M}_q$
given by the $q$ different values in the
sequence $\bar{M}$. 
In this case, we only have $\tilde{M}_1 = n$.
Then, we define the subsets $S_j = \{s_{\alpha}, \forall \alpha: M_{\alpha} = \tilde{M}_j\}$;
in this case we only have $S_1 = \{s_1, s_2\}$.
As said before, different bases are T-isomorphic if all sets $S_j$ are the same,
and we only need to keep one of them. 
Therefore, for the case of one 4-level system, we only have to keep the following bases,
which generate the representative subgroups:
\begin{align*}
\{\vec{e}_1, \vec{e}_2\} ,\; \{\vec{e}_1, 2 \vec{e}_2\} ,\; \{\vec{e}_1\} ,\; \{2 \vec{e}_1, 2 \vec{e}_2\}, \; \{2 \vec{e}_1 \},\; \{\}.\\
\end{align*}
\end{itemize}

Once again, automorphisms are characterized by $2 \times 2$ matrices 
$t_{\alpha \beta}$.
Since $M_1 = M_2 =n$, all $t_{\alpha \beta}$ fall
into the first case of Appendix \ref{app:b},
which implies that all $t_{\alpha \beta}$ are numbers modulo $p^{M_{\alpha}} = p^n$.
This gives a total of $p^{4n}$ possible matrices,
of which we only keep those that are invertible 
(have non-zero determinant modulo $p$).
For example, in the case of a 4-level system, there are $96$ such matrices.\\

\textbf{Find all subgroups:} As before, to find all subgroups of 
$\mathbb{Z}_d \oplus \mathbb{Z}_d$, we apply all
automorphisms to each of the representative subgroups found
in the first step and omit duplicates.
As always, these subgroups describe the indexes $\tau(m,n)$ which can 
have norm $1$. We show in Fig. \ref{fig:q4bit:weyl:erasing} some Weyl erasing
channels of a 4-level system that are completely characterized by subgroups 
of $\bbZ_4 \oplus \bbZ_4$.

\begin{figure}
\centering
\includegraphics[width=\columnwidth]{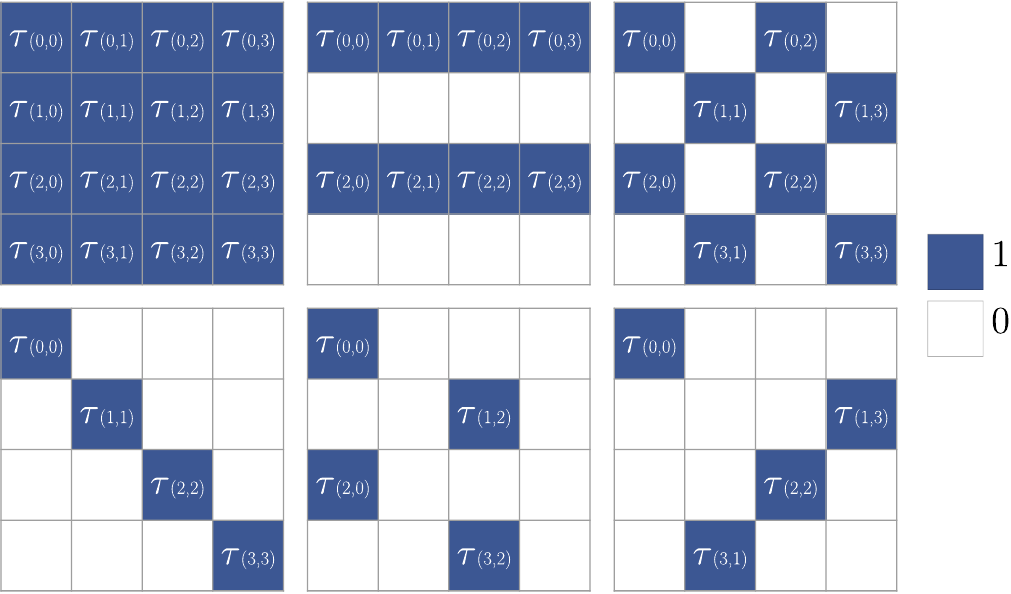} 
\caption{
Some 4-level system Weyl erasing channels 
with $\tau(m,n)=0,1$. Each of those is completely 
characterized by a set $\{ (m,n) : \tau(m,n)=1 \}$, which 
is a subgroup of $\bbZ_4 \oplus \bbZ_4$.
}
\label{fig:q4bit:weyl:erasing}
\end{figure}

We find all homomorphisms $\phi: \mathbb{Z}_{p^n} \oplus \mathbb{Z}_{p^n} \rightarrow \mathbb{Z}_{p^n}$.
As for the last case, the homomorphism is characterized by two values
$\phi_1 = \phi(\vec{e}_1), \phi_2 = \phi(\vec{e}_2)$.
Using Appendix \ref{app:c}, 
we find that $\phi_1$ and $\phi_2$ are both numbers modulo $p^{n}$.
To find all possible homomorphisms, we determine all possible pairs of $\phi_1,\phi_2$ 
which gives a total total of $p^{2n}$ homomorphisms. 
For the case of a 4-level system, the $16$ homomorphisms are
given by all pairs of numbers $\phi_1,\phi_2$  modulo $4$.
We show in Fig. \ref{fig:q4bit:weyl:erasing:general} some Weyl erasing 
channels for a 4-level system.

\begin{figure}
\includegraphics[width=\columnwidth]{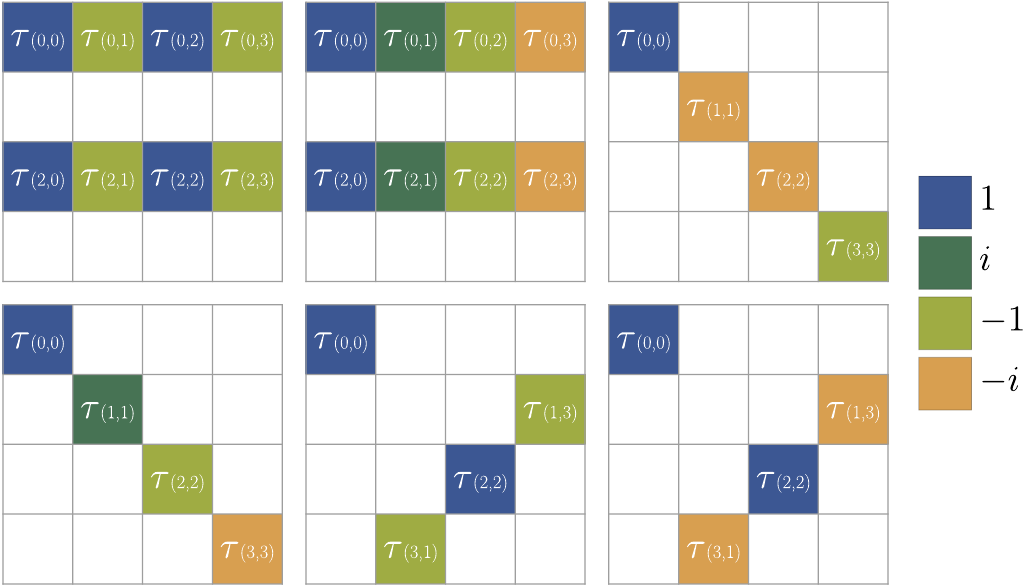} 
\caption{
Some Weyl erasing channels, with $d=4$, $N=1$, and $\abs{\tau(m,n)}=0,1$. 
Each of those is completely 
characterized by (1) a set $\{ (m,n) : \abs{\tau(m,n)}=1 \}$, which 
is a subgroup of $\bbZ_4 \oplus \bbZ_4$, and (2) an homomorphism 
$\phi : \bbZ_4 \oplus \bbZ_4 \mapsto \bbZ_4$.
}
\label{fig:q4bit:weyl:erasing:general}
\end{figure}

\subsection{Single particle with arbitrary dimension}
\label{sec: Single particle with arbitrary dimension}

Now we consider a single particle with arbitrary dimension $d$, 
which can be written with its prime factorization as $d=\prod_{i=1}^K p_i^{n_i}$.
In this case, what we have done in the last examples does not apply,
since it only applies for groups of the form  $\mathbb{Z}_{p^{M_1}} \oplus \mathbb{Z}_{p^{M_2}} \oplus \cdots \oplus \mathbb{Z}_{p^{M_K}}$(notice that all the groups in the sum are powers of the same prime). 

However, we can still find the subgroups of $\mathbb{Z}_d \oplus \mathbb{Z}_d$. 
To do it, we use the fact that $\mathbb{Z}_{pq} \simeq \mathbb{Z}_p \oplus \mathbb{Z}_q$ 
whenever $p$ and $q$ are coprime. 
Therefore, $\mathbb{Z}_d \simeq \mathbb{Z}_{p_1^{n_1}} \oplus \cdots \oplus \mathbb{Z}_{p_K^{n_K}}$,
and after reordering we have that:
\begin{equation}
\label{eq: decomposing d}
\mathbb{Z}_d \oplus \mathbb{Z}_d \simeq \bigoplus_{i=1}^K \mathbb{Z}_{p_i^{n_i}} \oplus \mathbb{Z}_{p_i^{n_i}}.
\end{equation}
Furthermore, it is a well known fact that subgroups of 
$F_1 \oplus F_2$ with $F_1$ and $F_2$ groups of coprime orders,
are obtained as cartesian products of subgroups of $F_1$ with subgroups of $F_2$. 
Therefore, because of the decomposition of eq.(\ref{eq: decomposing d}),
we can find the subgroups of $\mathbb{Z}_d \oplus \mathbb{Z}_d$
by obtaining all the subgroups of each $\mathbb{Z}_{p_i^{n_i}} \oplus \mathbb{Z}_{p_i^{n_i}}$ 
(which can be done as in the last example)
and then taking all their possible cartesian products. 

To find the homomorphisms $\phi: \mathbb{Z}_d \oplus \mathbb{Z}_d \rightarrow  \mathbb{Z}_d$,
we picture the $\phi$ as going from
$\bigoplus_{i=1}^K \mathbb{Z}_{p_i^{n_i}} \oplus \mathbb{Z}_{p_i^{n_i}}$
to $ \bigoplus_{i=1}^K \mathbb{Z}_{p_i^{n_i}}$. 
Any such homomorphism can be written as the
direct sums of homomorphisms $\phi_i: \mathbb{Z}_{p_i^{n_i}} \oplus \mathbb{Z}_{p_i^{n_i}} \rightarrow \mathbb{Z}_{p_i^{n_i}}$, 
which we obtained in the last example. 
Therefore, by constructing all such direct sums, we obtain all the homomorphisms we were looking for.

For example, if $d=12$ all we need to do is find
all the subgroups and homomorphisms of $\mathbb{Z}_4 \oplus  \mathbb{Z}_4$
and of $\mathbb{Z}_3 \oplus \mathbb{Z}_3$
and then take cartesian products of these subgroups and
the direct sum of the homomorphisms.

\subsection{N particles of dimension of prime power dimension}
\label{sec:N particles}
Now we consider a system consisting of $N$ particles,
each with dimension $p^{n_i}$ for $i =1, \cdots, N$,
ordered such that $n_1 \geq n_2 \geq \cdots \geq n_N$
(notice that the prime $p$ is the same for all particles).
In this case, the problem is to find all the subgroups of 
$\mathcal{G} = \mathbb{Z}_{p^{n_1}} \oplus \mathbb{Z}_{p^{n_1}} \oplus
\cdots \mathbb{Z}_{p^{n_N}} \oplus \mathbb{Z}_{p^{n_N}}$ 
and homomorphisms from $\mathcal{G}$ to $\mathbb{Z}_{p^{n_1}}$.
As an example, we will develop a system of one qubit and one 4-level system. \\

Similarly to the other examples, to find the representative subgroups we
take the following steps:
\begin{itemize}
\item Select a basis of $\mathcal{G}$. 
For example, in the case of a qubit and a 4-level system, the group is 
$\mathcal{G}= \mathbb{Z}_4 \oplus \mathbb{Z}_4 \oplus \mathbb{Z}_2 \oplus \mathbb{Z}_2$,
and we can choose a basis $\{\vec{e}_1, \vec{e}_2, \vec{e}_3, \vec{e}_4\}$, with
$
\vec{e}_1 = (1,0,0,0) \quad \vec{e}_2 = (0,1,0,0) \quad \vec{e}_3 = (0,0,1,0) \quad \vec{e}_4 = (0,0,0,1),$ 
where the first two entries add mod $4$ and the last two add mod $2$.
\item Next, we find the partition of $\mathcal{G}$. For the qubit and 4-level system, 
the orders of $\vec{e}_1$ and $\vec{e}_2$ are $4$ and the orders
of $\vec{e}_3, \vec{e}_4$ are $2$, so that
the partition of the group is $\bar{M} = M_1M_2M_3M_4 = 2211$.
\item We find all the sets $\mathbb{S} = \{s_{\alpha}\}$ with $0 \leq s_{\alpha} \leq M_{\alpha}$,
in this case there are $36$ said sets.
\item For each set $\mathbb{S}$, we define the basis $\mathbb{B} = \{p^{s_1} \vec{e}_1, p^{s_2} \vec{e}_2, p^{s_3} \vec{e}_3, p^{s_4} \vec{e}_4\}$.

\item As before, some of the bases created this way are redundant, 
since they are T-isomorphic. To eliminate this redundancy,
we first define $\tilde{M}_1, \cdots, \tilde{M}_q$ given 
by the $q$ different values of numbers in $\bar{M}$.
In the example of
a 4-level system and a qubit, we have that $\tilde{M}_1 = 2, \tilde{M}_2 = 1$.
Then, we define the sets $S_j = \{s_{\alpha}, \forall \alpha: M_{\alpha} = \tilde{M}_j\}$, 
which in this case are $S_1 = \{s_1, s_2\}$ and $S_2 = \{s_3, s_4\}$.
Finally, bases are $T-$isomorphic if their corresponding sets $S_j$ are equal. For example, the bases that come from the sets $\mathbb{S} = \{2,1,1,0\}$ and $\mathbb{S'} = \{1,2,0,1\}$
are $T-$isomorphic, since $S_1 = S_1' = \{2,1\}$ and $S_2 = S_2' = \{1,0\}$.
Therefore, after eliminating redundant bases and keeping only one of each
batch, 
we get the following $18$ bases:

\begin{widetext}
\begin{align*}
\mathbb{S} = \{0,0,0,0\} \rightarrow \mathcal{B}= \{\vec{e}_1, \vec{e}_2, \vec{e}_3, \vec{e}_4\}\; , \; \mathbb{S} = \{0,0,0,1\} \rightarrow \mathcal{B}=  \{\vec{e}_1, \vec{e}_2, \vec{e}_3\}\;, \; \mathbb{S} = \{0,0,1,1\}\rightarrow \mathcal{B}=  \{\vec{e}_1, \vec{e}_2\}, \\
\mathbb{S} = \{0,1,0,0\} \rightarrow \mathcal{B}=  \{\vec{e}_1, 2\vec{e}_2, \vec{e}_3, \vec{e}_4\}\;, \; \mathbb{S} = \{0,1,0,1\}\rightarrow \mathcal{B}=  \{\vec{e}_1, 2 \vec{e}_2, \vec{e}_3\}\;, \; \mathbb{S} = \{0,1,1,1\}\rightarrow \mathcal{B}=  \{\vec{e}_1, 2 \vec{e}_2\},\\
\mathbb{S} = \{0,2,0,0\}\rightarrow \mathcal{B}=  \{\vec{e}_1, \vec{e}_3, \vec{e}_4\}\;,\;
\mathbb{S} = \{0,2,0,1\} \rightarrow \mathcal{B}=  \{\vec{e}_1, \vec{e}_4\}\;, \; \mathbb{S} = \{0,2,1,1\}\rightarrow \mathcal{B}=  \{\vec{e}_1\},\\
\mathbb{S} = \{1,1,0,0\} \rightarrow \mathcal{B}=  \{2\vec{e}_1, 2 \vec{e}_2, \vec{e}_3, \vec{e}_4 \}\;, \; \mathbb{S} = \{1,1,0,1\}\rightarrow \mathcal{B}=  \{2 \vec{e}_1, 2 \vec{e}_2, \vec{e}_3\}\;, \; \mathbb{S} = \{1,1,1,1\}\rightarrow \mathcal{B}=  \{2 \vec{e}_1, 2 \vec{e}_2\},\\
\mathbb{S} = \{2,1,0,0\}\rightarrow \mathcal{B}=  \{2 \vec{e}_2, \vec{e}_3, \vec{e}_4\}\;, \; \mathbb{S} = \{2,1,0,1\}\rightarrow \mathcal{B}=  \{2 \vec{e}_2, \vec{e}_3\}\;, \; \mathbb{S} = \{2,1,1,1\}: \{2 \vec{e}_2\},\\
\mathbb{S} = \{2,2,0,0\}\rightarrow \mathcal{B}=  \{\vec{e}_3, \vec{e}_4\}\;, \; \mathbb{S} = \{2,2,0,1\}\rightarrow \mathcal{B}=  \{\vec{e}_3\}\;, \; \mathbb{S} = \{2,2,0,0\}\rightarrow \mathcal{B}=  \{\}
\end{align*}
\end{widetext}
As in the other cases, these bases form the 
representative subgroups of the group.

\end{itemize}
As before, the automorphisms are described by matrices $t_{\alpha \beta}$. 
For the special case of a qubit and 4-level system, the matrices
are of dimensions $4\times 4$ (because there are $4$ elements in the basis) 
and the conditions on the entries $t_{\alpha \beta}$ 
can be found using the cases described in Appendix \ref{app:b},  which lead to:
\begin{itemize}
\item $t_{11}$: $M_1 = M_1$ so that $t_{11}$ is a number modulo $p^{M_1} = 2^2 = 4$.
\item $t_{12}$: $M_1 = M_2$ so that $t_{12}$ is a number modulo $p^{M_1} = 2^2 = 4$.
\item $t_{13}$: $M_1 >  M_3$ so that $t_{13} = p^{M_{1}-M_{3}} \tau_{1 3} = 2 \tau_{13} $  with $\tau_{13}$ a number modulo $p^{M_3} = 2$. Therefore, the possible values are $0$ and $2$.
\item  The same can be done for the rest of the values, and we find that $t_{11},t_{12}, t_{21},t_{22}  \in \{0,1,2,3\}$; $t_{13}, t_{14}, t_{23}, t_{24} \in \{0,2\}$ and $t_{31}, t_{32}, t_{33}, t_{34}, t_{41}, t_{42}, t_{43}, t_{44} \in \{0,1\}$.
\end{itemize}
Then, running through all possible matrices with these entries and keeping only the invertible ones, we find $147456$ matrices. \\

As before, to find all subgroups, we apply these automorphisms to every representative subgroup and
discard repetitions. 
This procedure gives us the $249$ subgroups of $\mathbb{Z}_4 \oplus \mathbb{Z}_4 \oplus \mathbb{Z}_2 \oplus \mathbb{Z}_2$, some of which are
shown in fig. (\ref{fig:WCEPs}).

\begin{figure}
\centering
\includegraphics[width=\columnwidth]{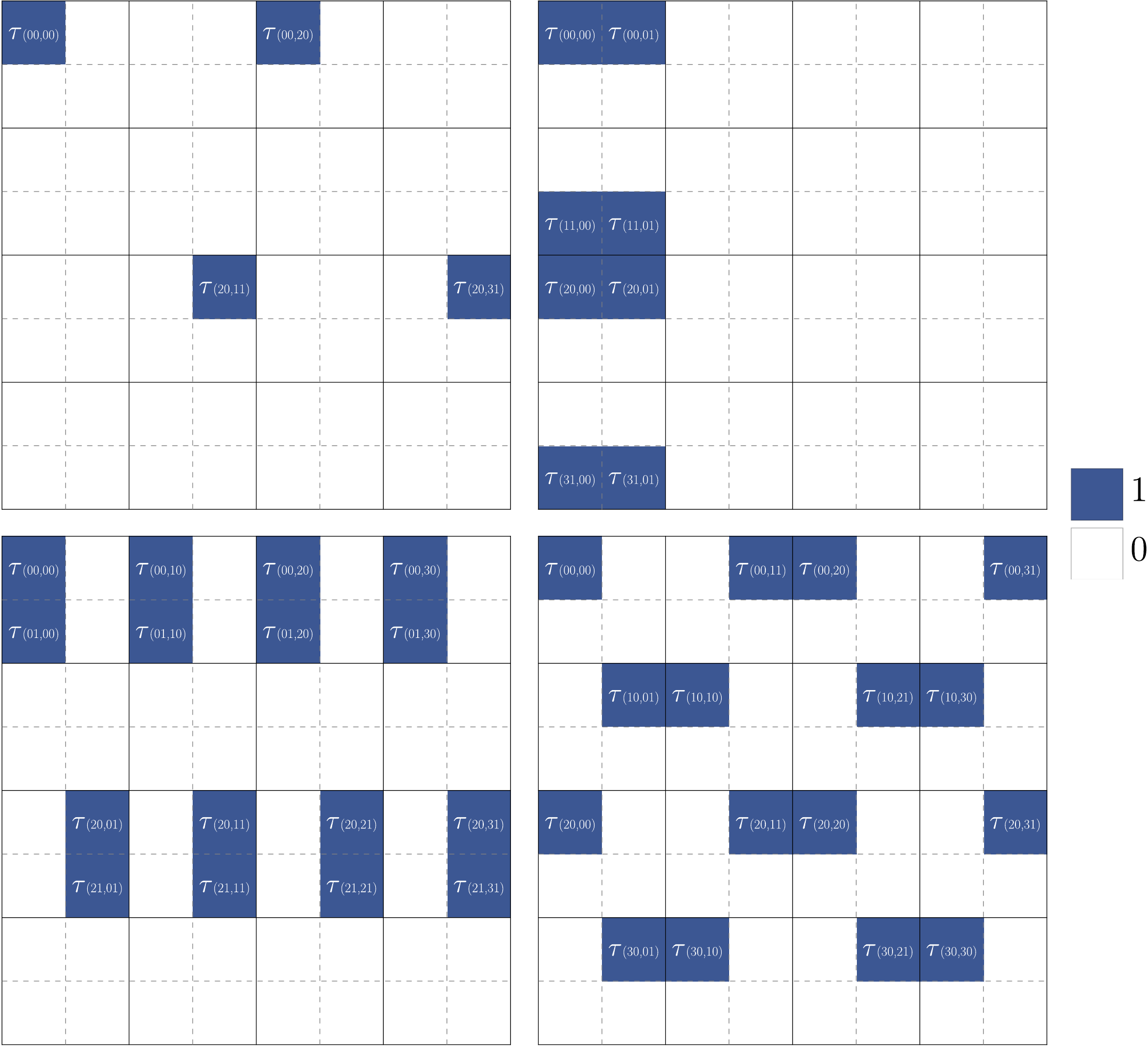}
\caption{
Some Weyl erasing channels of a system with a 4-level particle and a qubit,
with $\tau(\vec m, \vec n)=0,1$. 
Recall that $(m_\alpha, n_\alpha)$ corresponds to $\alpha$-th particle.
}
\label{fig:WCEPs}
\end{figure}

Finally, we find the homomorphisms $\phi: \mathcal{G} \rightarrow \mathbb{Z}_{p^{n_1}}$. 
For the case of a qubit and a 4-level system, we need the
homomorphisms $\phi: \mathbb{Z}_4 \oplus \mathbb{Z}_4 \oplus \mathbb{Z}_2 \oplus \mathbb{Z}_2 \rightarrow \mathbb{Z}_4$.
As before, we need to follow the procedure mentioned in Appendix \ref{app:c}. 
In this case, $n=2$ and therefore $\nu = 2$.
The homomorphisms $\phi$ are characterized by 
the values in the basis $\phi_{\alpha} = \phi(\vec{e}_{\alpha})$
which have to follow the conditions of eq.(\ref{eq:c1b}), that lead to:
\begin{itemize}
\item $\phi_1$: Since $\alpha= 1< 2= \nu$, we are in the second case of eq.(\ref{eq:c1b}), thus $\phi_1$ is a number modulo $p^n = 4$.
\item $\phi_2$: Since $\alpha= 2 =2= \nu$, we are in the first case of \eref{eq:c1b}, thus $\phi_2 = p^{n-M_2} s_2= s_2$ with $s_2$ a number modulo $p^{M_2} = 4$.
\item $\phi_3:$ Since $\alpha = 3 > 2$, $\phi_3 = p^{n-M_3} s_3 = 2s_3$ with $s_3$ a number modulo $p^{M_3} = 2$, so that $\phi_3 = 0,2$.
\item $\phi_4:$ Equivalently to $\phi_3$ we find that $\phi_4 =0,2$.
\end{itemize}
Therefore, the homomorphisms for a qubit and a 4-level system
are given by the $4$ numbers $\phi_1,\phi_2,\phi_3,\phi_4$ with $\phi_1,\phi_2 \in\{0,1,2,3\}$ and $\phi_3,\phi_4 \in \{0,2\}$ 
for a total of $64$ possibilities.
\subsection{Most General Case}
In the most general case
we have $N$ particles, each with arbitrary dimension $d_i$
and so the group under consideration is $\mathcal{G} = \bigoplus_{i=1}^N \mathbb{Z}_{d_i} \oplus \mathbb{Z}_{d_i}$.

Then, in this direct sum, we can first separate each $\mathbb{Z}_{d_i}$ as a sum of cyclic 
groups of prime power orders,
such as it was done in section 
\ref{sec: Single particle with arbitrary dimension} of this appendix.
Then, having written $\mathcal{G}$ as a direct sum
of cyclic groups with prime power order, we collect together
the cyclic groups of order that is a power of $2$, 
then cyclic groups of order power of $3$, $5$, $7$, and so on for each prime. 
 
After this, we can find the subgroups and homomorphisms of each of these
collections as it was done in 
section \ref{sec:N particles}. 
Finally, the subgroups of $\mathcal{G}$ can be found as cartesian products of subgroups 
of different collections.

\section{Number of subgroups per type $\overline{L}$}
\label{app:NumberOfSubgroups}

An expression for the number of different subgroups of type $\lpart$ is 
already known in the literature.
To introduce this expression we first need to consider the Ferrers graph of 
$\lpart$, that is, $L$ squares of which the first $L_1$ are in the first row, 
the next $L_2$ in the second, and so on. 
Then, the conjugate partition $\lpart^\prime$ is defined
as the Ferrers graph of $\lpart$ obtained by inverting rows and columns. 
Similarly, the partition $\mpart^\prime$ is defined as the conjugate partition
of $\lpart$. The number of subgroups $\mcH$ of type $\lpart$ of $\mcG_p$ is
given by
\begin{equation}
	\prod_{\alpha\geq 1}p^{M_{\alpha+1}^\prime(L_\alpha^\prime-M_\alpha^\prime)}
	\left[\genfrac{}{}{0pt}{0}{L_\alpha^\prime-M_{\alpha+1}^\prime}{M_\alpha^\prime-M_{\alpha+1}^\prime}
	\right]_p,
	\label{eq:border10}
\end{equation}
where the symbol
\begin{equation}
	\left[\genfrac{}{}{0pt}{0}{n}{m}
	\right]_p=\prod_{s=1}^m\frac{p^{n-s+1}-1}{p^{m-s+1}-1}
	\label{eq:border11}
\end{equation}
denotes the number of vector subspaces of dimension $m$ in a vector space of dimension
$n$ over the field $\mathbb{Z}_p$. The proof is rather intricate, and we refer the
reader to the relevant literature, such as \cite{Birkhoff1935, Butler1987}. 
However, the key fact is that the number of subgroups obtained by our algorithm can be 
compared with \eqref{eq:border10} to check that all subgroups of a given 
partition $\lpart$ have been found.
\bibliography{Bibliography}

\begin{thebibliography}{34}%
\makeatletter
\providecommand \@ifxundefined [1]{%
 \@ifx{#1\undefined}
}%
\providecommand \@ifnum [1]{%
 \ifnum #1\expandafter \@firstoftwo
 \else \expandafter \@secondoftwo
 \fi
}%
\providecommand \@ifx [1]{%
 \ifx #1\expandafter \@firstoftwo
 \else \expandafter \@secondoftwo
 \fi
}%
\providecommand \natexlab [1]{#1}%
\providecommand \enquote  [1]{``#1''}%
\providecommand \bibnamefont  [1]{#1}%
\providecommand \bibfnamefont [1]{#1}%
\providecommand \citenamefont [1]{#1}%
\providecommand \href@noop [0]{\@secondoftwo}%
\providecommand \href [0]{\begingroup \@sanitize@url \@href}%
\providecommand \@href[1]{\@@startlink{#1}\@@href}%
\providecommand \@@href[1]{\endgroup#1\@@endlink}%
\providecommand \@sanitize@url [0]{\catcode `\\12\catcode `\$12\catcode
  `\&12\catcode `\#12\catcode `\^12\catcode `\_12\catcode `\%12\relax}%
\providecommand \@@startlink[1]{}%
\providecommand \@@endlink[0]{}%
\providecommand \url  [0]{\begingroup\@sanitize@url \@url }%
\providecommand \@url [1]{\endgroup\@href {#1}{\urlprefix }}%
\providecommand \urlprefix  [0]{URL }%
\providecommand \Eprint [0]{\href }%
\providecommand \doibase [0]{https://doi.org/}%
\providecommand \selectlanguage [0]{\@gobble}%
\providecommand \bibinfo  [0]{\@secondoftwo}%
\providecommand \bibfield  [0]{\@secondoftwo}%
\providecommand \translation [1]{[#1]}%
\providecommand \BibitemOpen [0]{}%
\providecommand \bibitemStop [0]{}%
\providecommand \bibitemNoStop [0]{.\EOS\space}%
\providecommand \EOS [0]{\spacefactor3000\relax}%
\providecommand \BibitemShut  [1]{\csname bibitem#1\endcsname}%
\let\auto@bib@innerbib\@empty
\bibitem [{\citenamefont {Breuer}\ and\ \citenamefont
  {Petruccione}(2007)}]{breuer2007theory}%
  \BibitemOpen
  \bibfield  {author} {\bibinfo {author} {\bibfnamefont {H.-P.}\ \bibnamefont
  {Breuer}}\ and\ \bibinfo {author} {\bibfnamefont {F.}~\bibnamefont
  {Petruccione}},\ }\href
  {https://doi.org/10.1093/acprof:oso/9780199213900.001.0001} {\emph {\bibinfo
  {title} {{The Theory of Open Quantum Systems}}}}\ (\bibinfo  {publisher}
  {Oxford University Press},\ \bibinfo {year} {2007})\BibitemShut {NoStop}%
\bibitem [{\citenamefont {Rivas}\ and\ \citenamefont
  {Huelga}(2012)}]{rivas2012open}%
  \BibitemOpen
  \bibfield  {author} {\bibinfo {author} {\bibfnamefont {A.}~\bibnamefont
  {Rivas}}\ and\ \bibinfo {author} {\bibfnamefont {S.~F.}\ \bibnamefont
  {Huelga}},\ }\href@noop {} {\emph {\bibinfo {title} {Open quantum
  systems}}},\ Vol.~\bibinfo {volume} {10}\ (\bibinfo  {publisher} {Springer},\
  \bibinfo {year} {2012})\BibitemShut {NoStop}%
\bibitem [{\citenamefont {Schlosshauer}(2004)}]{Schlosshauer2004}%
  \BibitemOpen
  \bibfield  {author} {\bibinfo {author} {\bibfnamefont {M.}~\bibnamefont
  {Schlosshauer}},\ }\href {https://doi.org/10.1103/RevModPhys.76.1267}
  {\bibinfo {title} {{Decoherence, the measurement problem, and interpretations
  of quantum mechanics}}} (\bibinfo {year} {2004}),\ \Eprint
  {https://arxiv.org/abs/0312059} {arXiv:0312059 [quant-ph]} \BibitemShut
  {NoStop}%
\bibitem [{\citenamefont {Zurek}(2003)}]{RevModPhys.75.715}%
  \BibitemOpen
  \bibfield  {author} {\bibinfo {author} {\bibfnamefont {W.~H.}\ \bibnamefont
  {Zurek}},\ }\href {https://doi.org/10.1103/RevModPhys.75.715} {\bibfield
  {journal} {\bibinfo  {journal} {Rev. Mod. Phys.}\ }\textbf {\bibinfo {volume}
  {75}},\ \bibinfo {pages} {715} (\bibinfo {year} {2003})}\BibitemShut
  {NoStop}%
\bibitem [{\citenamefont {Nielsen}\ and\ \citenamefont
  {Chuang}(2010)}]{nielsen_chuang_2010}%
  \BibitemOpen
  \bibfield  {author} {\bibinfo {author} {\bibfnamefont {M.~A.}\ \bibnamefont
  {Nielsen}}\ and\ \bibinfo {author} {\bibfnamefont {I.~L.}\ \bibnamefont
  {Chuang}},\ }\href {https://doi.org/10.1017/CBO9780511976667} {\emph
  {\bibinfo {title} {Quantum Computation and Quantum Information: 10th
  Anniversary Edition}}}\ (\bibinfo  {publisher} {Cambridge University Press},\
  \bibinfo {year} {2010})\BibitemShut {NoStop}%
\bibitem [{\citenamefont {Wilde}(2017)}]{wilde2017}%
  \BibitemOpen
  \bibfield  {author} {\bibinfo {author} {\bibfnamefont {M.~M.}\ \bibnamefont
  {Wilde}},\ }\href {https://doi.org/10.1017/9781316809976} {\emph {\bibinfo
  {title} {Quantum Information Theory}}},\ \bibinfo {edition} {2nd}\ ed.\
  (\bibinfo  {publisher} {Cambridge University Press},\ \bibinfo {year}
  {2017})\BibitemShut {NoStop}%
\bibitem [{\citenamefont {Wolf}\ \emph {et~al.}(2008)\citenamefont {Wolf},
  \citenamefont {Eisert}, \citenamefont {Cubitt},\ and\ \citenamefont
  {Cirac}}]{Wolf2008assesing}%
  \BibitemOpen
  \bibfield  {author} {\bibinfo {author} {\bibfnamefont {M.~M.}\ \bibnamefont
  {Wolf}}, \bibinfo {author} {\bibfnamefont {J.}~\bibnamefont {Eisert}},
  \bibinfo {author} {\bibfnamefont {T.~S.}\ \bibnamefont {Cubitt}},\ and\
  \bibinfo {author} {\bibfnamefont {J.~I.}\ \bibnamefont {Cirac}},\ }\href
  {https://doi.org/10.1103/PhysRevLett.101.150402} {\bibfield  {journal}
  {\bibinfo  {journal} {Phys. Rev. Lett.}\ }\textbf {\bibinfo {volume} {101}},\
  \bibinfo {pages} {150402} (\bibinfo {year} {2008})}\BibitemShut {NoStop}%
\bibitem [{\citenamefont {Davalos}\ \emph {et~al.}(2019)\citenamefont
  {Davalos}, \citenamefont {Ziman},\ and\ \citenamefont
  {Pineda}}]{Davalos2019}%
  \BibitemOpen
  \bibfield  {author} {\bibinfo {author} {\bibfnamefont {D.}~\bibnamefont
  {Davalos}}, \bibinfo {author} {\bibfnamefont {M.}~\bibnamefont {Ziman}},\
  and\ \bibinfo {author} {\bibfnamefont {C.}~\bibnamefont {Pineda}},\ }\href
  {https://doi.org/10.22331/q-2019-05-20-144} {\bibfield  {journal} {\bibinfo
  {journal} {{Quantum}}\ }\textbf {\bibinfo {volume} {3}},\ \bibinfo {pages}
  {144} (\bibinfo {year} {2019})}\BibitemShut {NoStop}%
\bibitem [{\citenamefont {Heinosaari}\ and\ \citenamefont
  {Ziman}(2011)}]{heinosaari2011mathematical}%
  \BibitemOpen
  \bibfield  {author} {\bibinfo {author} {\bibfnamefont {T.}~\bibnamefont
  {Heinosaari}}\ and\ \bibinfo {author} {\bibfnamefont {M.}~\bibnamefont
  {Ziman}},\ }\href@noop {} {\emph {\bibinfo {title} {The mathematical language
  of quantum theory: from uncertainty to entanglement}}}\ (\bibinfo
  {publisher} {Cambridge University Press},\ \bibinfo {year}
  {2011})\BibitemShut {NoStop}%
\bibitem [{\citenamefont {Rivas}\ \emph {et~al.}(2014)\citenamefont {Rivas},
  \citenamefont {Huelga},\ and\ \citenamefont {Plenio}}]{Rivas2014}%
  \BibitemOpen
  \bibfield  {author} {\bibinfo {author} {\bibfnamefont {{\'{A}}.}~\bibnamefont
  {Rivas}}, \bibinfo {author} {\bibfnamefont {S.~F.}\ \bibnamefont {Huelga}},\
  and\ \bibinfo {author} {\bibfnamefont {M.~B.}\ \bibnamefont {Plenio}},\
  }\href {https://doi.org/10.1088/0034-4885/77/9/094001} {\bibfield  {journal}
  {\bibinfo  {journal} {Reports on Progress in Physics}\ }\textbf {\bibinfo
  {volume} {77}},\ \bibinfo {pages} {094001} (\bibinfo {year}
  {2014})}\BibitemShut {NoStop}%
\bibitem [{\citenamefont {Breuer}\ \emph {et~al.}(2016)\citenamefont {Breuer},
  \citenamefont {Laine}, \citenamefont {Piilo},\ and\ \citenamefont
  {Vacchini}}]{Breuer2016}%
  \BibitemOpen
  \bibfield  {author} {\bibinfo {author} {\bibfnamefont {H.~P.}\ \bibnamefont
  {Breuer}}, \bibinfo {author} {\bibfnamefont {E.~M.}\ \bibnamefont {Laine}},
  \bibinfo {author} {\bibfnamefont {J.}~\bibnamefont {Piilo}},\ and\ \bibinfo
  {author} {\bibfnamefont {B.}~\bibnamefont {Vacchini}},\ }\href
  {https://doi.org/10.1103/RevModPhys.88.021002} {\bibfield  {journal}
  {\bibinfo  {journal} {Reviews of Modern Physics}\ }\textbf {\bibinfo {volume}
  {88}},\ \bibinfo {pages} {021002} (\bibinfo {year} {2016})}\BibitemShut
  {NoStop}%
\bibitem [{\citenamefont {Gyongyosi}\ \emph {et~al.}(2018)\citenamefont
  {Gyongyosi}, \citenamefont {Imre},\ and\ \citenamefont
  {Nguyen}}]{Gyongyosi2018}%
  \BibitemOpen
  \bibfield  {author} {\bibinfo {author} {\bibfnamefont {L.}~\bibnamefont
  {Gyongyosi}}, \bibinfo {author} {\bibfnamefont {S.}~\bibnamefont {Imre}},\
  and\ \bibinfo {author} {\bibfnamefont {H.~V.}\ \bibnamefont {Nguyen}},\
  }\href {https://doi.org/10.1109/COMST.2017.2786748} {\bibfield  {journal}
  {\bibinfo  {journal} {IEEE Communications Surveys \& Tutorials}\ }\textbf
  {\bibinfo {volume} {20}},\ \bibinfo {pages} {1149} (\bibinfo {year}
  {2018})}\BibitemShut {NoStop}%
\bibitem [{\citenamefont {de~Leon}\ \emph {et~al.}(2022)\citenamefont
  {de~Leon}, \citenamefont {Fonseca}, \citenamefont {Leyvraz}, \citenamefont
  {Davalos},\ and\ \citenamefont {Pineda}}]{DeLeon2022}%
  \BibitemOpen
  \bibfield  {author} {\bibinfo {author} {\bibfnamefont {J.~A.}\ \bibnamefont
  {de~Leon}}, \bibinfo {author} {\bibfnamefont {A.}~\bibnamefont {Fonseca}},
  \bibinfo {author} {\bibfnamefont {F.}~\bibnamefont {Leyvraz}}, \bibinfo
  {author} {\bibfnamefont {D.}~\bibnamefont {Davalos}},\ and\ \bibinfo {author}
  {\bibfnamefont {C.}~\bibnamefont {Pineda}},\ }\href
  {https://doi.org/10.1103/PhysRevA.106.042604} {\bibfield  {journal} {\bibinfo
   {journal} {Physical Review A}\ }\textbf {\bibinfo {volume} {106}},\ \bibinfo
  {pages} {42604} (\bibinfo {year} {2022})},\ \Eprint
  {https://arxiv.org/abs/2205.05808} {arXiv:2205.05808} \BibitemShut {NoStop}%
\bibitem [{\citenamefont {Bru{\ss}}\ and\ \citenamefont
  {Macchiavello}(2002)}]{Brus2002}%
  \BibitemOpen
  \bibfield  {author} {\bibinfo {author} {\bibfnamefont {D.}~\bibnamefont
  {Bru{\ss}}}\ and\ \bibinfo {author} {\bibfnamefont {C.}~\bibnamefont
  {Macchiavello}},\ }\href {https://doi.org/10.1103/PhysRevLett.88.127901}
  {\bibfield  {journal} {\bibinfo  {journal} {Physical Review Letters}\
  }\textbf {\bibinfo {volume} {88}},\ \bibinfo {pages} {127901} (\bibinfo
  {year} {2002})}\BibitemShut {NoStop}%
\bibitem [{\citenamefont {Cerf}\ \emph {et~al.}(2002)\citenamefont {Cerf},
  \citenamefont {Bourennane}, \citenamefont {Karlsson},\ and\ \citenamefont
  {Gisin}}]{Cerf2002}%
  \BibitemOpen
  \bibfield  {author} {\bibinfo {author} {\bibfnamefont {N.~J.}\ \bibnamefont
  {Cerf}}, \bibinfo {author} {\bibfnamefont {M.}~\bibnamefont {Bourennane}},
  \bibinfo {author} {\bibfnamefont {A.}~\bibnamefont {Karlsson}},\ and\
  \bibinfo {author} {\bibfnamefont {N.}~\bibnamefont {Gisin}},\ }\href
  {https://doi.org/10.1103/PhysRevLett.88.127902} {\bibfield  {journal}
  {\bibinfo  {journal} {Physical Review Letters}\ }\textbf {\bibinfo {volume}
  {88}},\ \bibinfo {pages} {127902} (\bibinfo {year} {2002})}\BibitemShut
  {NoStop}%
\bibitem [{\citenamefont {Ralph}\ \emph {et~al.}(2007)\citenamefont {Ralph},
  \citenamefont {Resch},\ and\ \citenamefont {Gilchrist}}]{Ralph2007}%
  \BibitemOpen
  \bibfield  {author} {\bibinfo {author} {\bibfnamefont {T.~C.}\ \bibnamefont
  {Ralph}}, \bibinfo {author} {\bibfnamefont {K.~J.}\ \bibnamefont {Resch}},\
  and\ \bibinfo {author} {\bibfnamefont {A.}~\bibnamefont {Gilchrist}},\ }\href
  {https://doi.org/10.1103/PhysRevA.75.022313} {\bibfield  {journal} {\bibinfo
  {journal} {Physical Review A - Atomic, Molecular, and Optical Physics}\
  }\textbf {\bibinfo {volume} {75}},\ \bibinfo {pages} {022313} (\bibinfo
  {year} {2007})},\ \Eprint {https://arxiv.org/abs/0806.0654} {arXiv:0806.0654}
  \BibitemShut {NoStop}%
\bibitem [{\citenamefont {Campbell}(2014)}]{Campbell2014}%
  \BibitemOpen
  \bibfield  {author} {\bibinfo {author} {\bibfnamefont {E.~T.}\ \bibnamefont
  {Campbell}},\ }\href {https://doi.org/10.1103/PhysRevLett.113.230501}
  {\bibfield  {journal} {\bibinfo  {journal} {Physical Review Letters}\
  }\textbf {\bibinfo {volume} {113}},\ \bibinfo {pages} {230501} (\bibinfo
  {year} {2014})}\BibitemShut {NoStop}%
\bibitem [{\citenamefont {Wang}\ \emph {et~al.}(2020)\citenamefont {Wang},
  \citenamefont {Hu}, \citenamefont {Sanders},\ and\ \citenamefont
  {Kais}}]{Wang2020}%
  \BibitemOpen
  \bibfield  {author} {\bibinfo {author} {\bibfnamefont {Y.}~\bibnamefont
  {Wang}}, \bibinfo {author} {\bibfnamefont {Z.}~\bibnamefont {Hu}}, \bibinfo
  {author} {\bibfnamefont {B.~C.}\ \bibnamefont {Sanders}},\ and\ \bibinfo
  {author} {\bibfnamefont {S.}~\bibnamefont {Kais}},\ }\href
  {https://doi.org/10.3389/fphy.2020.589504} {\bibfield  {journal} {\bibinfo
  {journal} {Frontiers in Physics}\ }\textbf {\bibinfo {volume} {8}},\ \bibinfo
  {pages} {1} (\bibinfo {year} {2020})},\ \Eprint
  {https://arxiv.org/abs/2008.00959} {arXiv:2008.00959} \BibitemShut {NoStop}%
\bibitem [{\citenamefont {V{\'{e}}rtesi}\ \emph {et~al.}(2010)\citenamefont
  {V{\'{e}}rtesi}, \citenamefont {Pironio},\ and\ \citenamefont
  {Brunner}}]{Vertesi2010}%
  \BibitemOpen
  \bibfield  {author} {\bibinfo {author} {\bibfnamefont {T.}~\bibnamefont
  {V{\'{e}}rtesi}}, \bibinfo {author} {\bibfnamefont {S.}~\bibnamefont
  {Pironio}},\ and\ \bibinfo {author} {\bibfnamefont {N.}~\bibnamefont
  {Brunner}},\ }\href {https://doi.org/10.1103/PhysRevLett.104.060401}
  {\bibfield  {journal} {\bibinfo  {journal} {Physical Review Letters}\
  }\textbf {\bibinfo {volume} {104}},\ \bibinfo {pages} {060401} (\bibinfo
  {year} {2010})},\ \Eprint {https://arxiv.org/abs/0909.3171} {arXiv:0909.3171}
  \BibitemShut {NoStop}%
\bibitem [{\citenamefont {Skrzypczyk}\ and\ \citenamefont
  {Cavalcanti}(2018)}]{Skrzypczyk2018}%
  \BibitemOpen
  \bibfield  {author} {\bibinfo {author} {\bibfnamefont {P.}~\bibnamefont
  {Skrzypczyk}}\ and\ \bibinfo {author} {\bibfnamefont {D.}~\bibnamefont
  {Cavalcanti}},\ }\href {https://doi.org/10.1103/PhysRevLett.120.260401}
  {\bibfield  {journal} {\bibinfo  {journal} {Physical Review Letters}\
  }\textbf {\bibinfo {volume} {120}},\ \bibinfo {pages} {260401} (\bibinfo
  {year} {2018})},\ \Eprint {https://arxiv.org/abs/1803.05199}
  {arXiv:1803.05199} \BibitemShut {NoStop}%
\bibitem [{\citenamefont {Weyl}(1927)}]{weyl1927quantenmechanik}%
  \BibitemOpen
  \bibfield  {author} {\bibinfo {author} {\bibfnamefont {H.}~\bibnamefont
  {Weyl}},\ }\href {https://doi.org/10.1007/BF02055756} {\bibfield  {journal}
  {\bibinfo  {journal} {Zeitschrift f{\"u}r Physik}\ }\textbf {\bibinfo
  {volume} {46}},\ \bibinfo {pages} {1} (\bibinfo {year} {1927})}\BibitemShut
  {NoStop}%
\bibitem [{\citenamefont {Bertlmann}\ and\ \citenamefont
  {Krammer}(2008)}]{Bertlmann2008}%
  \BibitemOpen
  \bibfield  {author} {\bibinfo {author} {\bibfnamefont {R.~A.}\ \bibnamefont
  {Bertlmann}}\ and\ \bibinfo {author} {\bibfnamefont {P.}~\bibnamefont
  {Krammer}},\ }\href {https://doi.org/10.1088/1751-8113/41/23/235303}
  {\bibfield  {journal} {\bibinfo  {journal} {Journal of Physics A:
  Mathematical and Theoretical}\ }\textbf {\bibinfo {volume} {41}},\ \bibinfo
  {pages} {235303} (\bibinfo {year} {2008})},\ \Eprint
  {https://arxiv.org/abs/0806.1174} {arXiv:0806.1174} \BibitemShut {NoStop}%
\bibitem [{\citenamefont {Nathanson}\ and\ \citenamefont
  {Ruskai}(2007)}]{ruskai}%
  \BibitemOpen
  \bibfield  {author} {\bibinfo {author} {\bibfnamefont {M.}~\bibnamefont
  {Nathanson}}\ and\ \bibinfo {author} {\bibfnamefont {M.~B.}\ \bibnamefont
  {Ruskai}},\ }\href {https://doi.org/10.1088/1751-8113/40/28/S22} {\bibfield
  {journal} {\bibinfo  {journal} {Journal of Physics A: Mathematical and
  Theoretical}\ }\textbf {\bibinfo {volume} {40}},\ \bibinfo {pages} {8171}
  (\bibinfo {year} {2007})}\BibitemShut {NoStop}%
\bibitem [{\citenamefont {Ohno}\ and\ \citenamefont {Petz}(2009)}]{Ohno2009}%
  \BibitemOpen
  \bibfield  {author} {\bibinfo {author} {\bibfnamefont {H.}~\bibnamefont
  {Ohno}}\ and\ \bibinfo {author} {\bibfnamefont {D.}~\bibnamefont {Petz}},\
  }\href {https://doi.org/10.1007/s10474-009-8171-5} {\bibfield  {journal}
  {\bibinfo  {journal} {Acta Mathematica Hungarica}\ }\textbf {\bibinfo
  {volume} {124}},\ \bibinfo {pages} {165} (\bibinfo {year}
  {2009})}\BibitemShut {NoStop}%
\bibitem [{\citenamefont {Chru{\'{s}}ci{\'{n}}ski}\ and\ \citenamefont
  {Siudzi{\'{n}}ska}(2016)}]{Chruscinski2016}%
  \BibitemOpen
  \bibfield  {author} {\bibinfo {author} {\bibfnamefont {D.}~\bibnamefont
  {Chru{\'{s}}ci{\'{n}}ski}}\ and\ \bibinfo {author} {\bibfnamefont
  {K.}~\bibnamefont {Siudzi{\'{n}}ska}},\ }\href
  {https://doi.org/10.1103/PhysRevA.94.022118} {\bibfield  {journal} {\bibinfo
  {journal} {Physical Review A}\ }\textbf {\bibinfo {volume} {94}},\ \bibinfo
  {pages} {022118} (\bibinfo {year} {2016})},\ \Eprint
  {https://arxiv.org/abs/1606.02616} {arXiv:1606.02616} \BibitemShut {NoStop}%
\bibitem [{\citenamefont {Chru{\'{s}}ci{\'{n}}ski}\ and\ \citenamefont
  {Wudarski}(2013)}]{Chruscinski2013}%
  \BibitemOpen
  \bibfield  {author} {\bibinfo {author} {\bibfnamefont {D.}~\bibnamefont
  {Chru{\'{s}}ci{\'{n}}ski}}\ and\ \bibinfo {author} {\bibfnamefont {F.~A.}\
  \bibnamefont {Wudarski}},\ }\href
  {https://doi.org/10.1016/j.physleta.2013.04.020} {\bibfield  {journal}
  {\bibinfo  {journal} {Physics Letters, Section A: General, Atomic and Solid
  State Physics}\ }\textbf {\bibinfo {volume} {377}},\ \bibinfo {pages} {1425}
  (\bibinfo {year} {2013})},\ \Eprint {https://arxiv.org/abs/1212.2029}
  {arXiv:1212.2029} \BibitemShut {NoStop}%
\bibitem [{\citenamefont {Chru{\'{s}}ci{\'{n}}ski}\ and\ \citenamefont
  {Wudarski}(2015)}]{Chruscinski2015}%
  \BibitemOpen
  \bibfield  {author} {\bibinfo {author} {\bibfnamefont {D.}~\bibnamefont
  {Chru{\'{s}}ci{\'{n}}ski}}\ and\ \bibinfo {author} {\bibfnamefont {F.~A.}\
  \bibnamefont {Wudarski}},\ }\href
  {https://doi.org/10.1103/PhysRevA.91.012104} {\bibfield  {journal} {\bibinfo
  {journal} {Physical Review A - Atomic, Molecular, and Optical Physics}\
  }\textbf {\bibinfo {volume} {91}},\ \bibinfo {pages} {012104} (\bibinfo
  {year} {2015})},\ \Eprint {https://arxiv.org/abs/1408.1792} {arXiv:1408.1792}
  \BibitemShut {NoStop}%
\bibitem [{\citenamefont {Bengtsson}\ and\ \citenamefont
  {Życzkowski}(2017)}]{bengtsson_zyczkowski_2017}%
  \BibitemOpen
  \bibfield  {author} {\bibinfo {author} {\bibfnamefont {I.}~\bibnamefont
  {Bengtsson}}\ and\ \bibinfo {author} {\bibfnamefont {K.}~\bibnamefont
  {Życzkowski}},\ }\href {https://doi.org/10.1017/9781139207010} {\emph
  {\bibinfo {title} {Geometry of Quantum States}}}\ (\bibinfo  {publisher}
  {Cambridge University Press},\ \bibinfo {year} {2017})\BibitemShut {NoStop}%
\bibitem [{\citenamefont {Dahl}(2010)}]{dahl2010introduction}%
  \BibitemOpen
  \bibfield  {author} {\bibinfo {author} {\bibfnamefont {G.}~\bibnamefont
  {Dahl}},\ }\href@noop {} {\bibinfo {title} {{An introduction to convexity}}}
  (\bibinfo {year} {2010})\BibitemShut {NoStop}%
\bibitem [{\citenamefont {Burnside}(1911)}]{Burnside1911}%
  \BibitemOpen
  \bibfield  {author} {\bibinfo {author} {\bibfnamefont {W.}~\bibnamefont
  {Burnside}},\ }\href@noop {} {\emph {\bibinfo {title} {{Theory of groups of
  finite order}}}}\ (\bibinfo  {publisher} {The University Press},\ \bibinfo
  {year} {1911})\BibitemShut {NoStop}%
\bibitem [{\citenamefont {Hillar}\ and\ \citenamefont
  {Rhea}(2007)}]{Hillar2007}%
  \BibitemOpen
  \bibfield  {author} {\bibinfo {author} {\bibfnamefont {C.}~\bibnamefont
  {Hillar}}\ and\ \bibinfo {author} {\bibfnamefont {D.}~\bibnamefont {Rhea}},\
  }\href@noop {} {\bibfield  {journal} {\bibinfo  {journal} {The American
  Mathematical Monthly}\ }\textbf {\bibinfo {volume} {114}},\ \bibinfo {pages}
  {917} (\bibinfo {year} {2007})}\BibitemShut {NoStop}%
\bibitem [{\citenamefont {Birkhoff}(1935)}]{Birkhoff1935}%
  \BibitemOpen
  \bibfield  {author} {\bibinfo {author} {\bibfnamefont {G.}~\bibnamefont
  {Birkhoff}},\ }\href@noop {} {\bibfield  {journal} {\bibinfo  {journal}
  {Proceedings of the London Mathematical Society}\ }\textbf {\bibinfo {volume}
  {2}},\ \bibinfo {pages} {385} (\bibinfo {year} {1935})}\BibitemShut {NoStop}%
\bibitem [{\citenamefont {Siewert}(2022)}]{siewert2022orthogonal}%
  \BibitemOpen
  \bibfield  {author} {\bibinfo {author} {\bibfnamefont {J.}~\bibnamefont
  {Siewert}},\ }\href {https://doi.org/10.1088/2399-6528/AC6F43} {\bibfield
  {journal} {\bibinfo  {journal} {Journal of Physics Communications}\ }\textbf
  {\bibinfo {volume} {6}},\ \bibinfo {pages} {055014} (\bibinfo {year}
  {2022})}\BibitemShut {NoStop}%
\bibitem [{\citenamefont {Butler}(1987)}]{Butler1987}%
  \BibitemOpen
  \bibfield  {author} {\bibinfo {author} {\bibfnamefont {L.}~\bibnamefont
  {Butler}},\ }\href@noop {} {\bibfield  {journal} {\bibinfo  {journal}
  {Proceedings of the American Mathematical Society}\ }\textbf {\bibinfo
  {volume} {101}},\ \bibinfo {pages} {771} (\bibinfo {year}
  {1987})}\BibitemShut {NoStop}%
\end{thebibliography}%
\end{document}